\documentclass[a4paper]{article}
\usepackage{fullpage}
\usepackage{authblk}

\usepackage{mathtools}
\usepackage{amsthm}
\usepackage{amssymb}
\usepackage{cite}
\usepackage{tikz}
\usepackage{graphicx}
\usepackage[basic,classfont=caps]{complexity}
\usepackage[T1]{fontenc}
\usepackage{enumerate}
\usepackage{paralist}
\usepackage{algorithm}
\usepackage{algorithmic}
\usepackage{xspace}
\usepackage{accents}
\usepackage{stmaryrd}
\usepackage{enumitem}
\usepackage{todonotes}
\usepackage{microtype}

\newcommand{\myparagraph}[1]{\par\smallskip\noindent\textbf{#1.}}

\makeatletter
\let\epsilon\varepsilon
\let\phi\varphi
\let\emptyset\varnothing
\let\rho\varrho
\makeatother

\renewclass{\P}{Ptime}
\renewclass{\DTIME}{Dtime}
\renewclass{\DSPACE}{Dspace}
\renewclass{\EXP}{EXPtime}
\newclass{\TWOEXP}{2EXPtime}
\renewclass{\EXPSPACE}{EXPspace}
\newclass{\ACK}{Ack}
\newclass{\FDTIME}{FDtime}
\renewclass{\AP}{APtime}
\renewclass{\PSPACE}{Pspace}

\usetikzlibrary{positioning,fit,arrows,automata,calc,shapes,decorations.pathmorphing} 
\tikzset{->,>=stealth',
shorten >=1pt,shorten <=1pt,
auto,node distance=1.5cm,
every loop/.style={looseness=6},
initial text={},
every state/.style={inner sep=0.2mm, minimum size=0.5cm},
el/.style={font=\scriptsize},
inner sep=1mm,
loopright/.style={loop,looseness=6,out=30, in=-30},
loopleft/.style={loop,looseness=6,out=210, in=150},
loopabove/.style={loop,looseness=6,out=120, in=60},
loopbelow/.style={loop,looseness=6,out=300, in=240},
ve/.style={rectangle,draw,inner sep=1mm},
va/.style={circle,draw,inner sep=0.6mm},
}

\newtheorem{theorem}{Theorem}
\newtheorem{lemma}{Lemma}
\newtheorem{proposition}[lemma]{Proposition}
\newtheorem{corollary}[lemma]{Corollary}
\theoremstyle{remark}

\theoremstyle{definition}
\newtheorem{example}{Example}

\newcommand{\maxweight}{w_{\max}}
\newcommand{\maxm}{\mu_{\max}}
\newcommand{\defeq}{:=}
\newcommand\inter[1]{\llbracket #1 \rrbracket}
\newcommand{\calM}{\mathcal{M}}
\newcommand{\calN}{\mathcal{N}}
\newcommand{\calC}{\mathcal{C}}
\newcommand{\calD}{\mathcal{D}}
\newcommand{\calE}{\mathcal{E}}
\newcommand{\calG}{\mathcal{G}}
\newcommand{\calT}{\mathcal{T}}
\newcommand{\calU}{\mathcal{U}}
\newcommand{\stratregret}[2]{\mathbf{reg}^{#1}(#2)}
\providecommand\lang{}
\renewcommand{\lang}[1]{\mathcal{L}_{#1}}
\providecommand\st{}
\renewcommand{\st}{\:\mid\:}
\newcommand{\pow}{\mathcal{P}}
\newcommand{\eve}{Eve\xspace}
\newcommand{\adam}{Adam\xspace}
\newcommand{\ie}{\textit{i.e.}\xspace}
\newcommand{\eg}{\textit{e.g.}\xspace}
\newcommand{\prop}[1]{\textbf{\textsf{P}$_{#1}$}}

\newcommand{\joker}{\mathsf{joker}}
\newcommand{\credit}{\mathsf{Cr}}
\newcommand{\el}{\mathsf{EL}}
\newcommand{\dom}{\mathsf{dom}}
\newcommand{\Val}{\mathsf{Val}}

\newcommand{\JG}{\mathrm{JG}}

\title{On Delay and Regret Determinization of Max-Plus Automata
\thanks{This work was partially supported by the ERC Starting grant 279499
	(inVEST), the ARC project Transform (F\'ed\'eration Wallonie-Bruxelles),
	and the Belgian FNRS CDR project Flare. E. Filiot is an F.R.S.-FNRS
	research associate, I. Jecker an F.R.S.-FNRS Aspirant fellow, and G. A.
	P\'erez an F.R.S.-FNRS Aspirant fellow and FWA post-doc fellow.}
}

\author[1]{Emmanuel Filiot}
\author[1]{Isma\"el Jecker}
\author[1,2]{Nathan Lhote}
\author[1]{Guillermo A. P\'erez}
\author[1]{Jean-Fran\c{c}ois Raskin}
\affil[1]{Universit\'{e} Libre de Bruxelles}
\affil[2]{Universit\'e de Bordeaux, LaBRI}
\affil[ ]{\{efiliot, ijecker, gperezme, jraskin\}@ulb.ac.be, nlhote@labri.fr}

\begin{document}

\maketitle

\begin{abstract}
    Decidability of the determinization problem for weighted automata over the
    semiring $(\mathbb{Z}\cup  \{ -\infty \}, \max,+)$, WA for short, is a
    long-standing open question. We propose two ways of approaching it by
    constraining the search space of deterministic WA: $k$-delay and
    $r$-regret. A WA $\calN$ is $k$-delay determinizable if there exists a
    deterministic automaton $\calD$ that defines the same function as $\calN$
    and for all words $\alpha$ in the language of $\calN$, the
    accepting run of $\calD$ on $\alpha$ is always at most $k$-away from a
    maximal accepting run of $\calN$ on $\alpha$. That is, along all prefixes of
    the same length, the absolute difference between the running sums of weights
    of the two runs is at most $k$.
    A WA $\calN$ is $r$-regret determinizable if for all words $\alpha$ in its
    language, its non-determinism can be resolved on the fly to construct a run
    of $\calN$ such that the absolute difference between its value and the value
    assigned to $\alpha$ by $\calN$ is at most $r$.

    We show that a WA is determinizable if and only if it is $k$-delay
    determinizable for some $k$. Hence deciding the existence of some $k$ is
    as difficult as the general determinization problem. When $k$ and $r$ are
    given as input, the $k$-delay and $r$-regret determinization problems are
    shown to be \EXP-complete. We also show that determining whether a
    WA is $r$-regret determinizable for some $r$ is in \EXP.
\end{abstract}

\section{Introduction}

\myparagraph{Weighted automata} Weighted automata generalize finite
automata with weights on transitions~\cite{Droste:2009:HWA:1667106}. 
They generalize word languages to partial functions from words to
values of a semiring. First introduced by Sch\"utzenberger and Chomsky in the 60s,
they have been studied for long~\cite{Droste:2009:HWA:1667106}, with
applications in natural language and image processing for
instance. More recently, they have found new applications in
computer-aided verification as a measure of system
quality through quantitative properties~\cite{cdh10}, and in
system synthesis, as objectives for quantitative games~\cite{conf/concur/FiliotGR12}. 
In this paper, we consider weighted automata $\calN$ over the semiring
$(\mathbb{Z}\cup  \{ -\infty \}, \max,+)$, and just
call them weighted automata (WA). The value of a run is the sum of the weights
occurring on its transitions, and the value of a word is the maximal value of all
its accepting runs.  Absent transitions have a weight of $-\infty$ and runs of
value $-\infty$ are considered non-accepting. This defines a partial function
denoted $\inter{\calN} : \Sigma^*\rightarrow \mathbb{Z}$ whose domain is denoted
by $\lang{\calN}$. 

\myparagraph{Determinization problem} Most of the good algorithmic properties of
finite automata do not transfer to WA. Notably, the
(quantitative) language inclusion $\inter{\mathcal{A}} \leq \inter{\mathcal{B}}$ is undecidable for
WA~\cite{Krob:1992:EPR} (see also \cite{dav16} and
\cite{conf/atva/AlmagorBK11} for different proofs based on reductions
from the halting problem for two-counter machines). This has triggered 
research on sub-classes or other formalisms for which this problem becomes
decidable~\cite{DBLP:conf/fsttcs/FiliotGR14,conf/concur/ChatterjeeDEHR10}.
This includes the class of deterministic WA (DWA, also known as sequential
WA in the literature), which are the WA whose underlying (unweighted)
automaton is deterministic. Another scenario where it is desirable to
have a DWA is the quantitative synthesis problem,  
undecidable even for unambiguous WA, yet decidable for
DWA~\cite{conf/concur/FiliotGR12}. However, and in contrast with finite unweighted
automata, WA are not determinizable in general. For instance, the function
which outputs the maximal value between the number of $a$'s and the
number of $b$'s in a word $\alpha \in\{a,b\}^*$ is not realizable with a DWA. 
This motivates the determinization problem: given a WA $\calN$, is it
determinizable? I.e. is there a DWA defining the same
(partial) function as $\calN$? 

The determinization problem for computational models is
fundamental in theoretical computer science. For WA in particular, it
is sometimes more natural (and at least exponentially more succinct)
to specify a (non-deterministic) WA, even if some equivalent DWA
exists. If the function is specified in a weighted logic equivalent to
WA, such as weighted MSO~\cite{journals/tcs/DrosteG07}, the
logic-to-automata transformation may construct a non-deterministic,
but determinizable, WA. However, despite many research efforts, the largest class
for which this problem is known to be decidable is the class of
polynomially ambiguous WA~\cite{kirsten_et_al:LIPIcs:2009:1850}, and
the decidability status for the full class of WA is a long-standing
open problem. Other contributions and approaches to the determinization problem
include the identification of sufficient conditions for
determinizability~\cite{Mohri_FiniteState_CoLing}, approximate determinizability
(for unambiguous WA) where the DWA is required to produce values 
at most $t$ times the value of the WA, for a given factor
$t$~\cite{journals/tcs/AminofKL13}, and (incomplete) approximation algorithms
when the weights are non-negative~\cite{journals/mst/ColcombetD16}.

\myparagraph{Bounded-delay \& regret determinizers} In this paper, we adopt
another approach that consists in constraining the class of DWA that can be
used for determinization. More precisely, we define a \emph{class} of
DWA $\mathfrak{C}$ as a function from WA to sets of DWA, and
say that a DWA $\calD$ is a $\mathfrak{C}$-determinizer of a WA $\calN$ if
\begin{inparaenum}[$(i)$]
\item $\calD \in \mathfrak{C}(\calN)$ and
\item $\inter{\calN} =
\inter{\calD}$.
\end{inparaenum}
Then, $\calN$ is said to be $\mathfrak{C}$-deteminizable if it admits a
${\mathfrak{C}}$-determinizer. If ${\bf DWA}$ denotes the function mapping any
WA to the whole set of DWA, then obviously the ${\bf DWA}$-determinization
problem is the general (open) determinization problem. In this paper, we consider
two restrictions. 

First, given a bound $k\in\mathbb{N}$, we look for the class of $k$-delay DWA
${\bf Del}_k$, which maintain a strong relation with the sequence of values
along some accepting run of the non-deterministic automaton. More
precisely, a DWA $\calD$ belongs to ${\bf Del}_k(\calN)$ if for all words $\alpha \in
\lang{\calD}$, there is an accepting run $\rho$ of $\calN$ with maximal value
such that the running sum of the prefixes of $\rho$ and the running sum of the
prefixes of the unique run $\rho_{\calD}$ of $\calD$ on $\alpha$ are constantly
close in the following sense: for all lengths $\ell$, the absolute value of the
difference of the value of the prefix of $\rho$ of length $\ell$ and
the value of the prefix of $\rho_\calD$ of length $\ell$ is at most $k$.
Then the ${\bf
Del}_k$-determinization problem amounts to deciding whether there exists
${\calD} \in {\bf Del}_k$ such that $\inter{\calN}=\inter{\calD}$.  And if $k$
is left unspecified, it amounts to decide whether there exists ${\calD} \in
\bigcup_{k \in {\mathbb{N}}} {\bf Del}_k(\calN)$ such that
$\inter{\calN}=\inter{\calD}$. We note ${\bf Del}$ the function mapping any WA
$\calN$ to $\bigcup_k {\bf Del}_k(\calN)$. We will show that the ${\bf
Del}_k$-determinization problem is  complete for \EXP, and the ${\bf
Del}$-determinization problem is equivalent to the general ({\it unconstrained})
determinization problem.

The notion of delay has been central in many works on automata with
outputs. For instance, it has been a key notion in transducer theory
(automata with word outputs) for the determinization and the
functionality problems \cite{journals/tcs/BealCPS03},  and the decomposition of finite-valued
transducers~\cite{sakarovitch_et_al:LIPIcs:2008:1324}. The notion of delay has been also used in the
theory of WA, for instance to give sufficient conditions for
determinizability~\cite{Mohri_FiniteState_CoLing} or for the decomposition of finite-valued group
automata~\cite{DBLP:conf/fsttcs/FiliotGR14}. 

\begin{example} Let $A = \{a,b\}$ and
    $k\in\mathbb{N}$. The left automaton of Fig.~\ref{fig:example2}
    maps any word in $AaA^*$ to $0$, and any word in $AbA^*$ to
    $1$. It is ${\bf Del}_k$-determinizable by the right automaton of
    Fig.~\ref{fig:example2}. After one step, the delay of the DWA is
    $k$ with both transitions of the left WA. After two or more steps,
    the delay is always $0$. It is not ${\bf Del}_j$-determinizable for
    any $j<k$. (A second example of a bounded-delay determinizable automaton is
    shown in Fig.~\ref{fig:example22}.)
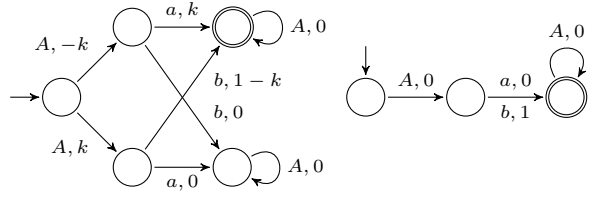
\begin{figure}[h]
\begin{center}
\begin{tikzpicture}[node distance=0.8cm]
	\node[state,initial left] (si) {};
	\node[state,above right= of si] (s0) {};
	\node[state,below right= of si] (s1) {};
	\node[state,right= of s0,accepting] (s2) {};
	\node[state,right= of s1] (s3) {};

	\path
	(si) edge node[el] {$A,-k$} (s0)
	(si) edge node[el,swap] {$A,k$} (s1)
	(s0) edge node[el] {$a,k$} (s2)
	(s1) edge node[el,swap,pos=0.8] {$b,1-k$} (s2)
	(s0) edge node[el,pos=0.8] {$b,0$} (s3)
	(s1) edge node[el,swap] {$a,0$} (s3)
	(s2) edge [loopright] node[el] {$A,0$} (s2)
	(s3) edge [loopright] node[el] {$A,0$} (s3)
	;

	\node[state,initial above] at (4,0) (qi) {};
	\node[state,right= of qi]  (q) {};
	\node[state,accepting, right= of q] (p) {};
        
        \path
	(qi) edge node[el] {$A,0$} (q)
	(q) edge node[el] {$a,0$} node[el,below] {$b,1$} (p)
	(p) edge [loopabove] node[el] {$A,0$} (p)
        ;
\end{tikzpicture}
\end{center}
\caption{A WA (left) and one of its $k$-delay determinizers (right)}
\label{fig:example2}
\end{figure}
\end{example}

\begin{figure}[h]
\begin{center}
\begin{tikzpicture}[node distance=0.8cm]
	\node[state,initial left] (si) {};
	\node[state,above right= of si,accepting] (s0) {};
	\node[state,below right= of si] (s1) {};
	\node[state,right=1.5cm of si,accepting] (s2) {};

	\path
	(si) edge node[el,align=left] {$a,k$\\$b,-k$} (s0)
	(si) edge node[el,swap,align=left] {$a,-k$\\$b,k$} (s1)
	(s0) edge node[el] {$a,0$} (s2)
	(s1) edge node[el,swap] {$b,0$} (s2)
	;

	\node[state,initial left,right=1cm of s2] (ti) {};
	\node[state,above right= of ti,accepting] (t0) {};
	\node[state,below right= of ti,accepting] (t1) {};
	\node[state,right=1.5cm of ti,accepting] (t2) {};

	\path
	(ti) edge node[el] {$a,k$} (t0)
	(ti) edge node[el,swap] {$b,-k$} (t1)
	(t0) edge node[el,align=left] {$a,0$\\$b,-2k$} (t2)
	(t1) edge node[el,swap,align=left] {$a,0$\\$b,2k$} (t2)
	;
\end{tikzpicture}
\end{center}
\caption{Another WA (left) and one of its $2k$-delay determinizers (right)}
\label{fig:example22}
\end{figure}
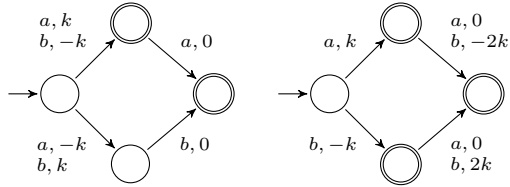

Second, we consider the class $\mathbf{Hom}$ of so-called homomorphic
DWA. Intuitively, any DWA $\calD\in \mathbf{Hom}(\calN)$
maintains a close relation with the structure of $\calN$: the
existence of an homomorphism from $\calD$ to $\calN$. An alternative definition
is that of a $0$-\emph{regret game} \cite{hpr15-journal} played on $\calN$:
\adam chooses input symbols one by one (forming a word $\alpha \in
\lang{\calN}$), while \eve reacts by choosing transitions of $\calN$, thus
constructing a run $\rho$ of ${\calN}$ on the fly (\ie without knowing the full
word $\alpha$ in advance). \eve wins the game if $\rho$ is accepting and its
value is equal to $\inter{\calN}(\alpha)$, \ie $\rho$ is a maximal accepting run
on $\alpha$. Then, any (finite memory) winning strategy for \eve can be seen as
a $\mathbf{Hom}$-determinizer of $\calN$ and conversely.  This generalizes the
notion of good-for-games automata, which do not need to be determinized prior to
being used as observers in a game, from the Boolean setting~\cite{hp06} to the
quantitative one. In some sense, $\mathbf{Hom}$-determinizable WA are ``good for
quantitative games'': when used as an observer in a quantitative game, Eve's
strategy can be applied on the fly instead of determinizing the WA and
constructing the synchronized product of the resulting DWA with the game arena.
This notion has been first introduced in~\cite{akl10} with motivations coming
from the analysis of online algorithms.  In~\cite{akl10},  it was shown that the
$\mathbf{Hom}$-determinization problem is in \P.

\begin{example}\label{ex:regret}
The following WA maps both $ab$ and $aa$ to $0$. It is not 
$\mathbf{Hom}$-determinizable because Eve has to choose whether to go
left or right on reading $a$. If she goes right, then Adam wins by
choosing letter $b$. If she goes left, Adam wins by picking $a$
again. However, it is \emph{almost} $\mathbf{Hom}$-determinizable by the DWA
obtained by removing the right part, in the sense that the function
realized by this DWA is $1$-close from the original one. This
motivates approximate determinization. 
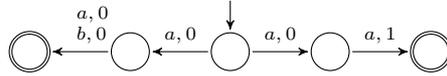
\begin{figure}[h]
\begin{center}
\begin{tikzpicture}[node distance=0.8cm]
	\node[state,initial above] (si) {};
	\node[state,right= of si] (s0) {};
	\node[state,left= of si] (s1) {};
	\node[state,right= of s0,accepting] (s2) {};
	\node[state,left= of s1,accepting] (s3) {};

	\path
	(si) edge node[el] {$a,0$} (s0)
	(si) edge node[el,swap] {$a,0$} (s1)
	(s0) edge node[el] {$a,1$} (s2)
	(s1) edge node[el,swap,align=left,pos=0.3] {$a,0$\\$b,0$} (s3)
	;
\end{tikzpicture}
\end{center}
\caption{A WA that is $\mathbf{Del}_0$-determinizable and 
$(1,\mathbf{Hom})$-determinizable, but not $\mathbf{Hom}$-determinizable}
\label{fig:example3}
\end{figure}

\end{example}

\myparagraph{Approximate determinization} Approximate
determinization of a WA $\calN$ relaxes the determinization problem
to determinizers which do not define
exactly the same function as $\calN$ but approximate it. Precisely, for a class
$\mathfrak{C}$ of DWA, $\calD$ a DWA and $r\in\mathbb{N}$,
we say that $\calD$ is an
$(r,\mathfrak{C})$-determinizer of $\calN$ if
\begin{inparaenum}[$(i)$]
\item $\calD\in \mathfrak{C}(\calN)$,
\item $\lang{\calD} = \lang{\calN}$ and
\item for all words $\alpha \in \lang{\calN}$, $|\inter{\calN}(\alpha)-\inter{\calD}(\alpha)|\leq r$.
\end{inparaenum}
Then, $\calN$ is $(r,\mathfrak{C})$-determinizable if
it admits some $(r,\mathfrak{C})$-determinizer, and it is 
\emph{approximately} $\mathfrak{C}$-determinizable if it is 
$(r,\mathfrak{C})$-determinizable for some $r$.

As Example~\ref{ex:regret} shows, there are WA that are approximately
$\mathbf{Hom}$-determinizable but not $\mathbf{Hom}$-determinizable, making
this notion appealing for the class of homomorphic determinizers. However,
there are classes $\mathfrak{C}$ for which a WA is approximately
$\mathfrak{C}$-determinizable if and only if it is
$\mathfrak{C}$-determinizable, making approximate determinization much less
interesting for such classes. This is the case for classes $\mathfrak{C}$ which
are \emph{complete for determinization} (Theorem~\ref{thm:rclose}), in the sense
that any determinizable WA is also $\mathfrak{C}$-determinizable. Obviously, the
class $\mathbf{DWA}$ is complete for determinization, but we show it is also the
case for the class of bounded-delay determinizers $\mathbf{Del}$
(Theorem~\ref{thm:semi-det}).
Therefore, we study approximate determinization for the class of homomorphic
determinizers only. We call such determinizers $r$-regret
determinizers, building on the regret game analogy given above. Indeed, a WA
$\calN$ is $(r,\mathbf{Hom})$-determinizable if and only if \eve wins the regret game
previously defined, with the following modified winning condition: the run that she
constructs on the fly must be such that
$|\inter{\calN}(\alpha)-\inter{\calD}(\alpha)|\leq r$, for all words $\alpha \in
\lang{\calN}$ that \adam can play.
We say that a WA $\calN$ is approximately $\mathbf{Hom}$-determinizable if
there exists a $r \in \mathbb{N}$ such that $\calN$ is
$(r,\mathbf{Hom})$-determinizable.

\myparagraph{Contributions} We show that the ${\bf Del}_k$-determinization problem
problem is \EXP-complete, even when $k$ is fixed
(Theorems~\ref{thm:kdelay-in-exp}~and~\ref{thm:kdelay-exp-hard}).
We also show that the class $\mathbf{Del}$ is
complete for determinization, \ie any determinizable WA $\calN$
is $k$-delay determinizable for some $k$ (Theorem~\ref{thm:semi-det}). This shows that
solving the $\mathbf{Del}$-determinization problem would solve the
(open) general determinization problem. This also gives a new
(complete) semi-algorithm for determinization, which consists in testing for the
existence of $k$-delay determinizers for increasing values $k$.  
We exhibit a family of bounded-delay determinizable WA, for delays
which depend exponentially on the WA. Despite our efforts, exponential delays
are the highest lower bound we have found. Interestingly, finding higher lower
bounds would lead to a better understanding of the determinization
problem, and proving that one of these lower bound is also an upper bound would
immediately give decidability. 
To decide ${\bf Del}_k$-determinization, we provide
a reduction to $\mathbf{Hom}$-determinization (\ie $0$-regret determinization),
which is known to be decidable in polynomial time~\cite{akl10}. 

We show that the approximate $\mathbf{Hom}$-determinization problem is
decidable in exponential time (Theorem~\ref{thm:existential-regret-prob}), a
problem which was left open in \cite{akl10}. This result is based on a
non-trivial extension to the quantitative setting of a game tool proposed by
Kuperberg and Skrzypczak in~\cite{ks15} for Boolean automata. In particular, our
quantitative extension is based on energy games~\cite{bflms08} while parity
games are sufficient for the Boolean case.  If $r$ is given (in
binary) the $(r,\mathbf{Hom})$-determinization problem is shown to be
\EXP-complete
(Theorems~\ref{thm:regret-upper}~and~\ref{thm:regret-lower-bound}). The hardness
holds even if $r$ is given in unary.  In the course of establishing our results,
we also show that every WA $A$ that is approximately
$\mathbf{Hom}$-determinizable is also {\em exactly} determinizable but there
may not be a homomorphism from a deterministic version of the automaton to the
original one (Lemma~\ref{lem:joker-nec-regret} and
Theorem~\ref{thm:win-joker-det}). Hence, the decision procedure for approximate
$\mathbf{Hom}$-determinizability can also be used as an algorithmically
verifiable sufficient condition for determinizability. 

\myparagraph{Other related works}
In transducer theory, a notion close to the notion of $k$-delay
determinizer has been introduced in \cite{fjlw16}, that of $k$-delay
uniformizers of a transducer. A uniformizer of a transducer $T$ is an
(input)-deterministic transducer such that the word-to-word function it defines
(seen as a binary relation) is included in the relation defined by
$T$, and any of its accepting runs should be $k$-delay close from some
accepting run of $T$. While the notion of $k$-delay uniformizer in
transducer theory is close to the notion of $k$-delay determinizer for
WA, the presence of a max operation in WA makes the $k$-delay
determinization problem conceptually harder. 

\section{Preliminaries}\label{sec:prelims}
We denote by $\mathbb{Z}$ the set of all integers; by $\mathbb{N}$, the set of
all non-negative integers, \ie the natural numbers---including $0$; by
$\mathbb{S}_{\lhd x}$, the subset $\{s \in \mathbb{S} \st s \lhd x\}$ of
any given set $\mathbb{S}$.
Finally, by $\epsilon$ we denote the \emph{empty word} over any alphabet.

\myparagraph{Automata}
A \emph{(non-deterministic weighted finite) automaton} $\calN = (Q, I, A,
\Delta,w, F)$ consists of a finite set $Q$ of
states, a set $I \subseteq Q$ of initial states,
a finite alphabet $A$ of symbols, 
a transition relation $\Delta
\subseteq Q \times A \times Q$, a weight function $w : \Delta \rightarrow
\mathbb{Z}$, and a set $F \subseteq Q$ of final states. By $\maxweight$ we
denote the maximal absolute value of a transition weight, \ie $\maxweight \defeq
\max_{\delta \in \Delta} |w(\delta)|$.
We say $\calN$ is \emph{pair-deterministic} if $|I| = 1$ and for all
$(q,a) \in Q \times A$ we have that $(q,a,q_1), (q,a,q_2) \in \Delta$ implies
$q_1 = q_2$ or $w(q,a,q_1) \neq w(q,a,q_2)$;
\emph{deterministic}, if $|I| = 1$ and
for all $(q,a) \in Q \times A$ we have that $(q,a,q_1), (q,a,q_2) \in \Delta$
implies $q_1 = q_2$.

A \emph{run} of $\calN$ on a word $a_0 \dots a_{n-1} \in A^*$ is a sequence
$\rho = q_0 a_0 q_1 \dots q_{n-1} a_{n-1} q_n \in (Q \cdot A)^*Q$ such that
$(q_i,a_i,q_{i+1}) \in \Delta$ for all $0 \le i < n$.
We say $\rho$ is \emph{initial} if $q_0 \in I$; \emph{final}, if $q_n \in F$;
\emph{accepting}, if it is both initial and final.
The automaton $\calN$ is said to be \emph{trim} if for all states $q \in Q$,
there is a run from a state $q_I \in I$ to $q$ and there is a run from $q$ to
some $q_F \in F$.
The \emph{value of
$\rho$}, denoted by $w(\rho)$, corresponds to the sum of the weights of its
transitions:
\(
	w(\rho) \defeq \sum_{i = 0}^{n-1} w(q_i,a_i,q_{i+1}).
\)

The automaton $\calN$ has the (unweighted) \emph{language}
$\lang{\calN} = \{ \alpha \in A^* \st$ there is an accepting run of
$\calN \text{ on } \alpha \}$ and defines a function $\inter{\calN}
: \lang{\calN} \to \mathbb{Z}$ as follows $\alpha \mapsto \max\{
w(\rho) \st \rho$ is an accepting run of $\calN$ on $\alpha\}$. A run $\rho$ of
$\calN$ on $\alpha$ is said to be \emph{maximal} if $w(\rho) =
\inter{\calN}(\alpha)$.

\myparagraph{Determinization with delay}
Given $k \in \mathbb{N}$ and two automata $\calN = (Q,I, A,
\Delta, w,F)$ and $\calN' =  (Q', I', A, \Delta', w',F')$, we say that
$\calN$ is \emph{$k$-delay-included} (or $k$-included, for short) in
$\calN'$, denoted by $\calN \subseteq_k \calN'$, if for every accepting run
$\rho = q_0 a_0 \dots a_{n-1} q_n$ of $\calN$, there exists an accepting
run $\rho' = q_0' a_0 \dots a_{n-1} q_n'$ of $\calN'$ such that $w'(\rho') =
w(\rho)$, and for every $1 \leq i \leq n$, $|w'(q_0' \dots q_i') - w(q_0 \dots
q_i)| \leq k$. For an automaton $\calN$, we denote by $\mathbf{Del}_k(\calN)$
the set $\{\calD \in \mathbf{DWA} \st 
\calD \subseteq_k \calN\}$.

An automaton $\calN$ is said to be \emph{$k$-delay determinizable} if there
exists an automaton $\calD \in \mathbf{Del}_k(\calN)$ such that $\inter{\calD} =
\inter{\calN}$. Such an automaton is called a \emph{$k$-delay determinizer of
$\calN$}.

\myparagraph{Determinization with regret}
Given two automata $\calN = (Q,I,A,\Delta,w,F)$ and $\calN' =
(Q',I',A,\Delta',w',F')$, a mapping $\mu : Q \to Q'$ from
states in $\calN$ to states in $\calN'$ is a \emph{homomorphism from $\calN$ to
$\calN'$} if $\mu(I) \subseteq I'$, $\mu(F) \subseteq F'$, $\{
\left(\mu(p),a,\mu(q) \right) \st (p,a,q) \in \Delta\} \subseteq \Delta'$, and
$w'(\mu(p),a,\mu(q)) = w(p,a,q)$. For an automaton $\calN$, we denote by
$\mathbf{Hom}(\calN)$ the set of deterministic automata $\calD$ for which there
is a homomorphism from $\calD$ to $\calN$.
The following lemma follows directly from the
preceding definitions.
\begin{lemma}\label{lem:hom-prop}
	For all automata $\calN$, for all $\calD \in \mathbf{Hom}(\calN)$, we
	have that $\lang{\calD} \subseteq \lang{\calN}$ and
	$\inter{\calD}(\alpha) \leq \inter{\calN}(\alpha)$ for all $\alpha \in
	\lang{\calN}$.
\end{lemma}

Given $r \in
\mathbb{N}$ and an automaton $\calN$, we say $\calN$ is \emph{$r$-regret
determinizable} if there is a deterministic automaton $\calD$ such that:
\begin{inparaenum}[$(i)$]
\item\label{itm:homo} $\calD \in \mathbf{Hom}(\calN)$,
	\item\label{itm:same-lang} $\lang{\calN} = \lang{\calD}$, and
	\item\label{itm:dist-r} $\sup_{\alpha \in \lang{\calN}}
		\left|\inter{\calN}(\alpha) -
		\inter{\calD}(\alpha)\right| \le r$.
\end{inparaenum}
The automaton $\calD$ is said to be an \emph{$r$-regret determinizer of
$\calN$}.  Note that~(\ref{itm:homo}) implies we can remove the absolute
value in~(\ref{itm:dist-r}) because of Lemma~\ref{lem:hom-prop}.
%
%

\myparagraph{Regret games}
Given $r \in \mathbb{N}$ and an automaton $\calN =
(Q,I,A,\Delta,w,F)$, an \emph{$r$-regret game} is a two-player turn-based game
played on $\calN$ by \eve and \adam. To begin, \eve chooses an initial state.
Then, the game proceeds in rounds as follows. From the current state $q$, \adam
chooses a symbol $a \in A$ and \eve chooses a new state $q'$ (not necessarily a
valid $a$-successor of $q$). After a word $\alpha \in \lang{\calN}$ has been
played by \adam, he may decide to stop the game. At this point \eve loses if the
current state is not final or if she has not constructed a valid run of $\calN$
on $\alpha$. Furthermore, she must pay a (\emph{regret}) value equal to
$\inter{\calN}(\alpha)$ minus the value of the run she has constructed.

Formally, a \emph{strategy for \adam} is a finite word $\alpha \in A^*$ from runs to
symbols and a \emph{strategy for \eve} is a function $\sigma : (Q \cdot A)^* \to
Q$ from state-symbol sequences to states. Given a word (strategy) $\alpha = a_0
\dots a_{n-1}$, we write $\sigma(\alpha)$ to denote the sequence $q_0 a_0 \dots
a_{n-1} q_n$ such that $\sigma(\epsilon) = q_0$ and $\sigma(q_0 a_0 \dots q_i
a_i) = q_{i+1}$ for all $0 \le i < n$. The \emph{regret of $\sigma$} is defined
as follows:
\( 
	\stratregret{\sigma}{\calN} \defeq \sup_{\alpha \in \lang{\calN}}
	\inter{\calN}(\alpha) - \Val(\sigma(\alpha))
\)
where, for all sequences $\rho \in (Q \cdot A)^*Q$, the function $\Val(\rho)$ is
such that $\rho \mapsto  w(\rho)$ if $\rho$ is an accepting run of $\calN$ and
$\rho \mapsto -\infty$ otherwise.  We say \eve wins the $r$-regret game played
on $\calN$ if she has a strategy such that $\stratregret{\sigma}{\calN} \le r$.
Such a strategy is said to be \emph{winning} for her in the regret game.

\myparagraph{Games \& determinization}
A \emph{finite-memory} 
strategy $\sigma$ for \eve in a regret game played on an automaton
$\calN = (Q,I,A,\Delta,w,F)$ is a strategy that can be encoded as a deterministic 
\emph{Mealy machine} $\calM = (S,s_I,A,\lambda_u,\lambda_o)$ where
$S$ is a finite set of (memory) states, $s_I$ is the initial state,
$\lambda_u : S \times A \to S$ is the update function and $\lambda_o : S \times (A
\cup \{\epsilon\}) \to Q$ is the output function. The machine encodes $\sigma$
in the following sense: $\sigma(\epsilon) = \lambda_o(s_I,\epsilon)$ and
$\sigma(q_0 a_0 \dots q_n a_n) = \lambda_o(s_{n}, a_{n})$
where $s_0 = s_I$ and $s_{i+1} = \lambda_u(s_i, a_i)$ for all $0 \le i < n$.  We
then say that $\calM$ \emph{realizes} the strategy $\sigma$ and that $\sigma$
has \emph{memory} $|S|$.  In particular, strategies which have memory $1$ are
said to be \emph{positional} (or \emph{memoryless}).

A finite-memory strategy $\sigma$ for \eve in a regret game played on $\calN$
defines the deterministic automaton $\calN_\sigma$ obtained by taking the
\emph{synchronized product} of $\calN$ and the finite Mealy machine
$(S,s_I,A,\lambda_u,\lambda_o)$ realizing $\sigma$. Formally $\calN_\sigma$
is the automaton $(Q\times S,(\lambda_o(s_I,\epsilon),s_I), A,\Delta',w',F \times S)$
where: $\Delta'$ is the set of all triples $\left( (q,s),a,(q',s') \right)$ such
that $(q,s) \in Q \times S,a \in A$, $s' = \lambda_u(s,a)$, and $q' =
\lambda_o(s,a)$; and $w'$ is such that $\left((q,s),a,(q',s')\right) \mapsto
w(q,a,q')$.

We remark that, for all $r \in \mathbb{N}$, for all finite-memory strategies
$\sigma$ for \eve such that $\stratregret{\sigma}{\calN} \le r$, we have that
$\calN_\sigma$ is an $r$-regret determinizer of $\calN$. Indeed, the desired
homomorphism from $\calN_\sigma$ to $\calN$ is the projection on the first
dimension of $Q \times S$, \ie $(q,s) \mapsto q$. Furthermore, from any
$r$-regret determinizer $\calD$ of $\calN$, it is straightforward to define a
finite-memory strategy for \eve that is winning for her in the $r$-regret game.

\begin{lemma}\label{lem:games-det}
	For all $r \in \mathbb{N}$, an automaton $\calN$ is $r$-regret
	determinizable if and only if there exists a finite-memory strategy
	$\sigma$ for \eve such that $\stratregret{\sigma}{\calN} \le r$.
\end{lemma}

In~\cite{akl10} it was shown that if there exists a $0$-regret strategy for \eve
in a regret game, then a $0$-regret memoryless strategy for her exists as well.
Furthermore, deciding if the latter holds is in \P. Hence,
by Lemma~\ref{lem:games-det} we obtain the following.

\begin{proposition}[From~\cite{akl10}]\label{pro:0regret-easy}
	Determining if a given automaton is $0$-regret determinizable is
	decidable in polynomial time.
\end{proposition}

%
%
%

\myparagraph{A sufficient condition for determinizability}
Given $B \in \mathbb{N}$, we say an automaton $\calN$ is $B$-bounded if it is
trim and for every maximal accepting run $\rho_p = p_0 a_0 p_1 \dots a_{n-1}
p_n$ of $\calN$,  for every $0 \leq i \leq n$, and for every initial run $\rho_q
= q_0 a_0 q_1 \dots a_{i-1} q_i$, we have $w(\rho_q) - w(p_0 a_0 p_1 \dots
a_{i-1} p_i) \leq B$.
	
We now prove that, given a $B$-bounded automaton, we are able to build an
equivalent deterministic automaton.

\begin{proposition}\label{pro:bound-to-det}
	Let $B \in \mathbb{N}$ and let $\calN = (Q,I,A,\Delta,w,F)$ be an automaton.
	If $\calN$ is $B$-bounded, then there exists a
	deterministic automaton $\calD$ such that $\inter{\calD} =
	\inter{\calN}$, and whose size and maximal weight are
	polynomial w.r.t. $\maxweight$ and $B$, and exponential w.r.t. $|Q|$.
\end{proposition}

\begin{proof}[Sketch]
	The result is proved by exposing the construction of the deterministic
	automaton $\calD$, inspired by the \emph{determinization algorithm}
	presented in~\cite{Mohri_FiniteState_CoLing}.  On each input word
	$\alpha \in A^*$, $\calD$ outputs the value of the maximal initial run
	$\rho_\alpha$ of $\calN$ on $\alpha$ (respectively the maximal
	accepting run if $\alpha \in \lang{\calN}$), and keeps track of all the
	other initial runs on $\alpha$ by storing in its state the pairs
	$(q,w_q) \in Q \times \{-B, \ldots, B\}$ such that the maximal initial
	run on $\alpha$ that ends in $q$ has weight $w(\rho_{\alpha}) + w_q$.
	If for some state $q$ the delay $w_q$ gets lower than $-B$, the
	$B$-boundedness assumption allows $\calD$ to drop the corresponding runs
	without modifying the function defined: whenever a run has a delay
	smaller than $-B$ with respect to $\rho_\alpha$, no continuation will
	ever be maximal. This ensures that our construction always yields a
	finite automaton, unlike the determinization algorithm, that does not
	always terminate.	
\end{proof}

\myparagraph{On complete-for-determinization classes}
Given $r \in \mathbb{N}$, a class $\mathfrak{C}$ of DWA, and an automaton
$\calN$, we say $\calN$ is $(r,\mathfrak{C})$-determinizable if there exists
$\calD \in \mathfrak{C}(\calN)$ such that:
\begin{inparaenum}[$(i)$]
	\item $\lang{\calN} = \lang{\calD}$, and
	\item $\sup_{\alpha \in \lang{\calN}}
		\left|\inter{\calN}(\alpha) -
		\inter{\calD}(\alpha)\right| \le r$.
\end{inparaenum}

We will now confirm our claim from the introduction: approximate determinization
is not interesting for some classes.

\begin{proposition}\label{pro:finite-range-bound}
	Let $\calN = (Q,I,A,\Delta,w,F)$ be a trim automaton such that the range
	of $\inter{\calN}$ is included into $\{-B, \ldots, B\}$, for some $B \in
	\mathbb{N}$. Then $\calN$ is determinizable.
\end{proposition}

\begin{proof}
	Let $\rho_p = p_0 a_0 p_1 \dots a_{n-1} p_n$ be a maximal accepting run
	of $\calN$ and let $\rho_q = q_0 a_0 q_1 \dots a_{i-1} q_i$ be an
	initial run of length $i\leq n$.  We define $\rho_p^i = p_0 a_0 p_1
	\dots a_{i-1} p_i$ for $i\leq n$.

	By assumption, we have $w(\rho_p)\geq -B$.  By trimness assumption, the
	state $q_i$ can reach a final state and we have
	$w(\rho_q)-|Q|\maxweight\leq B$ otherwise there would be an accepting
	run of value greater than $B$.  Similarly, since state $p_i$ can be
	reached from an initial state, we have $-|Q|\maxweight +a \leq B$, with
	$a=w( p_i a_i \dots a_{n-1} p_n)=w(\rho_p)-w(\rho_p^i)$.  By combining
	the three constraints, we obtain: $w(\rho_q)-|Q|\maxweight
	-|Q|\maxweight +a -w(\rho_p)\leq 3B$ which, once rearranged, yields:
	$w(\rho_q) -w(\rho_p^i) \leq 3B + 2|Q|\maxweight$.  Thus, $\calN$ is
	$(3B + 2|Q|\maxweight)$-bounded and determinizable (by
	Proposition~\ref{pro:bound-to-det}).
\end{proof}

Recall that a class $\mathfrak{C}$ of DWA is complete for determinization if any
determinizable automaton is also $\mathfrak{C}$-determinizable.
\begin{theorem}\label{thm:rclose}
	Given a complete-for-determinization class $\mathfrak{C}$ of DWA, an
	automaton $\calN$ is $(r,\mathfrak{C})$-determinizable, for some $r \in
	\mathbb{N}$, if and only if it is $\mathfrak{C}$-determinizable.
\end{theorem}
\begin{proof}
	If $\calN$ is determinizable, then in particular it is
	$(r,\mathfrak{C})$-determinizable for any $r$.  Conversely, let us
	assume that $\calD$ is an $(r,\mathfrak{C})$-determinizer of
	$\calN$, for some $r$.

	Then one can construct an automaton $\calM$ such that $\inter \calM
	=\inter\calN -\inter\calD$ by taking the product of $\calN$ and $\calD$
	with transitions weighted by the difference of the weights of $\calN$
	and $\calD$. Since $\calD$ is $r$-close to $\calN$, the range of $\calM$
	is included in the set $\{ -r,\ldots ,r \}$.  This means, according to
	Propositions~\ref{pro:finite-range-bound},
	that $\calM$ (once trimmed) is
	determinizable and that one can construct a deterministic automaton
	realizing $\inter\calD +\inter\calM=\inter \calN$. Since
	$\mathfrak{C}$ is complete for determinization, the result follows.
\end{proof}

\section{Deciding $k$-delay determinizability}
\label{sec:kdelay}
In this section we prove that deciding $k$-delay determinizability is
\EXP-complete.
First, however, we show that $k$-delay determinization is complete for
determinization: if a given automaton is determinizable, then there is a $k$
such that it is $k$-delay determinizable as well. Hence, exposing an upper bound
for $k$ would lead to an algorithm for the general determinizability problem. We
also give a family of automata for which an exponential delay is required.

\subsection{Completeness for determinization}
\label{sec:semialgo}

\begin{theorem}\label{thm:semi-det}
	If an automaton $\calN$ is determinizable, then there exists $k
	\in \mathbb{N}$ such that $\calN$ is $k$-delay determinizable.
\end{theorem}
\begin{proof}
	We proceed by contradiction. Suppose $\calN = (Q,I,A,\Delta,w,F)$ is
	determinizable. Denote by $\calD = (Q',\{q_I'\},A,\Delta',w',F')$ a
	deterministic automaton such that $\inter{\calD} =
	\inter{\calN}$. Let us assume, towards a contradiction,
	that for all $k \in \mathbb{N}$ there is
	no deterministic automaton $\calE$ such that $\calE \subseteq_k \calN$
	and $\inter{\calD} = \inter{\calE}$.
	In particular,
	we have that $\calD \not\subseteq_\chi \calN$ for $\chi \defeq
	|Q||Q'|(\maxweight + \maxweight')$. This means that there is a word
	$\alpha = a_0 \dots a_{n-1} \in \lang{\calN}$ such that for a maximal
	accepting run $\rho = q_0 a_0 \dots a_{n-1} q_n$ of $\calN$ on $\alpha$
	it holds that
	\begin{equation}\label{eqn:bigdiff-abs}
		\left| w(q_0 a_0\dots a_{\ell-1} q_\ell) -
		w'(q'_0 a_0 \dots a_{\ell-1} q'_\ell) \right| > \chi
	\end{equation}
	for some $0 \le \ell \le n$ and $q'_0 a_0 \dots a_{n-1} q'_n$ the
	unique initial run of $\calD$ on $\alpha$.
	We consider the two possibilities.
	
	Suppose that $w(q_0 a_0\dots a_{\ell-1} q_\ell) - w'(q'_0 a_0 \dots
	a_{\ell-1} q'_\ell) > \chi$.  Then, at least one final state $q_n$ is
	reachable in $\calN$ from $q_\ell$, and the shortest path to it consists
	of at most $|Q|$ transitions. Since $\chi \ge |Q|(\maxweight +
	\maxweight')$, $\calD$ does not realize the same function as $\calN$,
	which contradicts our hypothesis. 
	
	Suppose that
	\(
		w'(q'_0 a_0 \dots a_{\ell-1} q'_\ell) -
		w(q_0 a_0 \dots a_{\ell-1} q_\ell) > \chi.
	\)
	Using the fact that $\chi = |Q||Q'|(\maxweight + \maxweight')$, we
	expose a loop that can be pumped down to present a word mapped to
	different values by $\calD$ and $\calN$.  For every $0 \leq j \leq
	|Q||Q'|$, let $0 \leq i_j \leq \ell$ denote the minimal integer
	satisfying $w'(q'_0  a_0 \dots a_{i_j-1}q'_{i_j}) - w(q_0 a_0 \dots
	a_{i_j-1}q_{i_j}) \geq j(\maxweight + \maxweight')$.  Then there exist
	$0 \leq j < k \leq |Q||Q'|$ such that $q_{i_j} = q_{i_k}$, $q'_{i_j} =
	q'_{i_k}$.  Moreover, since $\maxweight$ and $\maxweight'$ correspond to
	the maximal weights of $\calN$ and $\calD$ respectively, $i_j \neq i_k$
	holds, and
	\(
		w(q_{i_j}a_{i_j} \dots a_{i_k-1}q_{i_k}) > w'(q'_{i_j} a_{i_j}
		\dots a_{i_k-1}q'_{i_k}).
	\)
	Since $\calD$ is deterministic, it assigns a strictly lower value than
	$\inter{\calN}$ to the word $a_0 \dots a_{i_j-1} a_{i_k} \dots a_{n-1}$,
	which contradicts our assumption that $\calD$ realizes the same function
	as $\calN$.
\end{proof}


Although we do not have an upper bound on the
$k$ needed for a determinizable automaton to be $k$-delay determinizable, we are
able to provide an exponentially large lower bound.

\begin{proposition}\label{pro:exp-delay-necessary}
	Given an automaton $\calN = (Q,I,A,\Delta,w,F)$, a delay $k$ as big as
	$2^{\mathcal{O}(|Q|)}$ might be needed for it to be $k$-delay
	determinizable.
\end{proposition}

To prove the above proposition we will make use of the language of words with a
\emph{$j$-pair}~\cite{kz15}.

\myparagraph{Words with a $j$-pair}
Consider the alphabet $A = \{1,\dots,n\}$.  Let $\alpha = a_0 a_1 \dots
\in A^*$ and $j \in A$. A \emph{$j$-pair} is a pair of positions
$i_1 < i_2$ such that $a_{i_1} = a_{i_2} = j$ and $a_k \le j$, for all
$i_1 \le k \le i_2$.

\begin{lemma}\label{lem:props-jpair}
	For all $j \in A$:
	\begin{inparaenum}[$(i)$]
		\item for all $\alpha \in A^*$, if $\alpha$ contains no
			$j$-pair, then $|\alpha| < 2^{n}$;
		\item for all $j \in A$, there exists $\alpha \in A^*$ such that
			$|\alpha| = 2^{n} - 1$ and $\alpha$ contains no
			$j$-pair.
	\end{inparaenum}
\end{lemma}
\begin{proof}
	A proof of the first claim is given by Klein and Zimmermman
	in~\cite{kz15} (Theorem $1$).

	To show the second claim holds as well, we can inductively construct a
	word with the desired property. As the base case, consider $\alpha_1 = 1$.
	Thus, for some $i$, there is $\alpha_i$ which contains no $j$-pair,
	contains no letter bigger than $i$, and is of length $2^i - 1$. For the
	inductive step, we let $\alpha_i = \alpha_{i-1} i \alpha_{i-1}$. It is
	easy to verify that the properties hold once more.
\end{proof}

We will now focus on the function $f : A^* \to \mathbb{Z}$, which maps a word
$\alpha$ to $0$ if it contains a $j$-pair and to $-|\alpha|$ otherwise.
Fig.~\ref{fig:length-jpair} depicts
the automaton $\calN$ realizing $f$ with $3n + 1$ states.
Proposition~\ref{pro:exp-delay-necessary} then follows from the following
result.

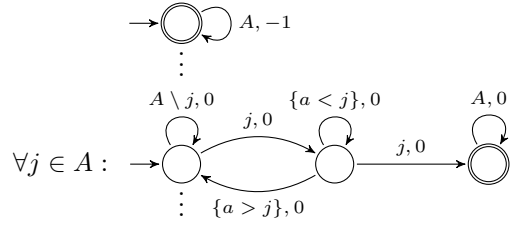
\begin{figure}
\begin{center}
\begin{tikzpicture}

	\node[state,initial left,accepting] (main) {};
	\node[below=-0.2cm of main] (dots1) {$\vdots$};
	\node[state,initial left,below=0.8cm of dots1,
		label={[label distance=0.6cm]left:$\forall j \in A:$}] (initj) {};
	\node[state,right= of initj] (midj) {};
	\node[state,right= of midj,accepting] (finalj) {};
	\node[below=-0.2cm of initj] (dots2) {$\vdots$};

	\path
	(main) edge[loopright] node[el] {$A,-1$} (main)
	(initj) edge[loopabove] node[el] {$A \setminus j, 0$} (initj)
	(initj) edge[bend left] node[el] {$j, 0$} (midj)
	(midj) edge[loopabove] node[el] {$\{a < j\},0$} (midj)
	(midj) edge[bend left] node[el] {$\{a > j\},0$} (initj)
	(midj) edge node[el] {$j,0$} (finalj)
	(finalj) edge[loopabove] node[el] {$A,0$} (finalj)
	;
\end{tikzpicture}
\caption{Automaton $\calN$ realizing function $f$ which outputs the
(negative) length of the word if it has no $j$-pair.}
\label{fig:length-jpair}
\end{center}
\end{figure}

\begin{lemma}\label{lem:high-delay}
	Any determinizer of automaton $\calN$ (see
	Fig.~\ref{fig:length-jpair}), which realizes the function $f$, has a
	delay of at least $2^{n}-1$.
\end{lemma}
\begin{proof}
	Consider a word $\alpha$ of length $2^{n} - 1$ containing no
	$j$-pair---which exists according to Lemma~\ref{lem:props-jpair}.
	Further consider an arbitrary $k$-determinizer $\calD$ for $\calN$. We
	remark that $\inter{\calD}(\alpha) = \inter{\calN}(\alpha) = 1 -
	2^{n}$, since both automata realize $f$ and $\alpha$ does not contain a
	$j$-pair. It follows from Lemma~\ref{lem:props-jpair} that $\alpha \cdot
	a$ contains a $j$-pair (that is, for all $a \in A$). Hence, for all $a
	\in A$, we have $\inter{\calN}(\alpha \cdot a) = 0$. Furthermore, by
	construction of $\calN$, for all maximal accepting runs $q_0 a_0 \dots
	a_{|\alpha| +1} q_{|\alpha| +2}$ of $\calN$ on $\alpha \cdot a$ we have
	$w(q_0 \dots q_i) = 0$ for all $1 \le i \le |\alpha| + 2$. In
	particular, for $i = |\alpha| + 1$, we have $\inter{\calD}(\alpha) -
	w(q_0 \dots q_i) = 2^{n} - 1$ which proves the claim.
\end{proof}

\subsection{Upper bound}
We now argue that $0$-delay determinizability is in \EXP.  Then, we show how to
reduce (in exponential time) $k$-delay determinizability to $0$-delay
determinizability. We claim that the composition of the two algorithms remains
singly exponential.
\begin{proposition}\label{pro:0delay-in-exp}
	Deciding the $0$-delay problem for a given automaton is in
	\EXP.
\end{proposition}

The result will follow from Propositions~\ref{pro:0regret-easy}
and~\ref{pro:zero-delay-regret-pdet}.
Before we state and prove Proposition~\ref{pro:zero-delay-regret-pdet} we need
some intermediate definitions and lemmas. The following properties of
$k$-inclusion, which follow directly from 
the definition, will be useful later.
%
\begin{lemma}\label{lem:k-inc_prop}
	For all automata $\calN$, $\calN'$, and $\calN''$, for all $k,k' \in
	\mathbb{N}$, the following hold:
	\begin{enumerate}[nolistsep]
		\item\label{lem-item:k-inc_prop_less} if $\calN \subseteq_k \calN'$ and $k \le k'$, then
			$\calN \subseteq_{k'} \calN'$;
		\item\label{lem-item:k-inc_prop_trans} if $\calN
			\subseteq_{k} \calN'$ and $\calN' \subseteq_{k'}
			\calN''$, then $\calN \subseteq_{k + k'} \calN''$;
		\item\label{lem-item:k-inc_prop_sub} if $\calN \subseteq_k
			\calN'$, then $\lang{\calN} \subseteq
			\lang{\calN'}$, and for every $\alpha \in
			\lang{\calN}$, $\inter{\calN}(\alpha) \leq
			\inter{\calN'}(\alpha)$.
	\end{enumerate}
\end{lemma}

We now show how to decide $0$-delay determinizability by reduction to
$0$-regret determinizability.  Let us first convince the reader that $0$-regret
determinizability implies $0$-delay determinizability.

\begin{proposition}\label{pro:zero-delay-regret}
	If an automaton $\calN$ is $0$-regret determinizable, then it is
	$0$-delay determinizable.
\end{proposition}
\begin{proof}
	We have, from Lemma~\ref{lem:games-det}, that \eve has a finite-memory
	winning strategy $\sigma$ in the $0$-regret game played on $\calN$.
	Then, by definition of the regret game, $\lang{\calN} =
	\lang{\calN_\sigma}$, and for every $\alpha \in
	\lang{\calN}$, $\inter{\calN}(\alpha) -
	\inter{\calN_\sigma}(\alpha) \leq 0$, hence $\inter{\calN} =
	\inter{\calN_\sigma}$. Moreover, as \eve chooses a run in
	$\calN$, we have $\calN_\sigma \subseteq_0 \calN$.  Therefore
	$\calN_\sigma$ is a $0$-delay determinizer of $\calN$.
\end{proof}

The converse of the above result does not hold in general (see
Fig.~\ref{fig:example3}). Nonetheless, it holds when the automaton is
pair-deterministic. We now show that, under this hypothesis, an automaton is
$0$-regret determinizable if and only if the automaton is $0$-delay
determinizable.

\begin{proposition}\label{pro:zero-delay-regret-pdet}
	A pair-deterministic automaton $\calN$ is $0$-delay
	determinizable if and only if it is $0$-regret determinizable.
\end{proposition}

\begin{proof}[Sketch]
	If $\calN$ is $0$-regret determinizable, then $\calN$ is $0$-delay
	determinizable by Proposition~\ref{pro:zero-delay-regret}. Now suppose
	that $\calN$ is $0$-delay determinizable, and let $\calD$ be a $0$-delay
	determinizer of $\calN$. For every initial run
	$\rho_{\alpha} = p_0 a_0 p_1 \ldots a_{n-1} p_{n}$ of $\calD$ on input
	$\alpha = a_0 \ldots a_{n-1}$, there exists exactly one initial run
	$\rho'_{\alpha} = p'_0 a_0 p'_1 \ldots a_{n-1} p'_{n}$ of $\calN$ such
	that for every $1 \leq i \leq n$, $w(p_0' \dots p_i') = w'(p_0 \dots
	p_i)$.  The existence of $\rho'_{\alpha}$ is guaranteed by the fact
	that $\calD$ is a $0$-delay determinizer of $\calN$, and, since $\calN$
	is pair-deterministic, such a run is unique. Then the strategy for \eve
	in the $0$-regret game played on $\calN$ obtained by following, given an
	input word $\alpha$, the run $\rho'_{\alpha}$ of $\calN$, is winning.
\end{proof}

We observe that any automaton $\calN = (Q,I,A,\Delta,w,F)$ can be transformed
into a pair-deterministic automaton $\pow(\calN)$ with at most an exponential
blow-up in the state-space.  Intuitively, we merge all the states from the
original automaton which can be reached by reading $a \in A$ and taking a
transition with weight $x \in \mathbb{Z}$. This is a generalization of the
classical subset construction used to determinize unweighted automata.
Critically, the construction is such that $\pow(\calN) \subseteq_0 \calN$ and
$\calN \subseteq_0 \pow(\calN)$. (For completeness, the construction is given in
appendix.) The next result then follows immediately from the latter property and
from Lemma \ref{lem:k-inc_prop} item~\ref{lem-item:k-inc_prop_trans}.

\begin{proposition}\label{pro:zero-delay-subset}
	An automaton $\calN$ is $0$-delay determinizable if and only if
	$\pow(\calN)$ is $0$-delay determinizable if and only if $\pow(\calN)$
	is $0$-regret determinizable.
\end{proposition}

We now show how to extend the above techniques to the general case of
$k$-delay.

\begin{theorem}\label{thm:kdelay-in-exp}
	Deciding the $k$-delay problem for a given automaton is in
	\EXP.
\end{theorem}

Given an automaton $\calN$ and $k \in \mathbb{N}$, we will construct a new automaton
$\delta_k(\calN)$ that will encode delays (up to $k$) in its state space. In this new
automaton, for every state-delay pair $(p,i)$ and for every transition
$(p,a,q) \in \Delta$, we will have an $a$-labelled transition to $(q,j)$ with
weight $i + w(p,a,q) - j$ for all $-k \le j \le k$. Intuitively, $i$ is the
amount of delay the automaton currently has, and to get to a point where the
delay becomes $j$ via transition $(p,a,q)$ a weight of $i + w(p,a,q) - j$ must
be outputted.  We will then show that the resulting automaton is $0$-delay
determinizable if and only if the original automaton is $k$-delay
determinizable.

\myparagraph{$k$-delay construction}
Let $\calN =  (Q, I, A, \Delta, w,F)$ be an automaton.
Let $\delta_k(\calN) =  (Q', I', A, \Delta',w',F')$ be the
automaton defined as follows.
\begin{itemize}[nolistsep]
	\item $Q' = Q \times \{ -k, \dots, k \}$;
	\item $I'= I \times \{0\}$;
	\item $\Delta' = \{ ((p,i), a, (q,j)) \st (p, a, q) \in \Delta \}$;
	\item $w' : \Delta' \rightarrow \mathbb{Z}$, $((p,i), a, (q,j))
		\mapsto i + w(p,a,q) - j$;
	\item $F' = F \times \{ 0 \}$.
\end{itemize}

\begin{lemma}\label{lem:delay-construction}
	The $k$-delay construction satisfies the following properties.
	\begin{enumerate}[nolistsep]
		\item \label{lem-item:del_sub} $\delta_k(\calN) \subseteq_k
			\calN$;
		\item \label{lem-item:del_sup} for every automaton
			$\calM$ such that $\calM \subseteq_{k}
			\calN$, we have $\calM \subseteq_0 \delta_k(\calN)$;
		\item \label{lem-item:del_equ} $\inter{\delta_k(\calN)} =
			\inter{\calN}$.
	\end{enumerate}
\end{lemma}
\begin{proof}
	\item \ref{lem-item:del_sub})
	Let
	$(q_0,i_0) a_0 (q_1,i_1) \dots a_{n-1}(q_n,i_n)$
	be an accepting run of $\delta_k(\calN)$.  Then
	$q_0 a_0 q_1 \dots a_{n-1} q_n$ is an accepting run of $\calN$,
	and for every $0 \leq j < n$,
	\begin{align*}
		& \textstyle \left|\sum_{\ell = 0}^{j}
			w'((q_{\ell},i_{\ell}),a_\ell,(q_{\ell+1},i_{\ell+1})) - \sum_{\ell =
			0}^{j} w(q_{\ell},a_\ell,q_{\ell+1})\right|\\
		= & \textstyle \left|\sum_{\ell = 0}^{j} (i_{\ell} + w(q_{\ell},a_\ell,q_{\ell+1})
			-i_{\ell+1} -w(q_{\ell},a_\ell,q_{\ell+1}))\right|\\
		= & |i_0 - i_j| = |i_j| \leq k \text{ (since $i_0 = 0$).}
	\end{align*}
	Therefore $\delta_k(\calN) \subseteq_k \calN$.

	\item \ref{lem-item:del_sup})
	Let $\calM =  (Q'', I'', A, \Delta'', w'',F'')$ be an
	automaton such that $\calM \subseteq_{k} \calN$.  For every
	accepting run $p_0a_0 \dots a_{n-1} p_n$ of
	$\calM$, there exists an accepting run
	$q_0 a_0 \dots a_{n-1} q_n$ of $\calN$ such that for
	every $0 \leq j < n$ 	
	\[ \textstyle
		\sum_{l = 0}^j \left(w(q_{l},a_l,q_{l+1}) - 
		w''(p_{l},a_l,p_{l+1})\right) \in \{ -k, \dots, k \}.
	\]
	Let $i_j$ denote the above value.
	Then
	$(q_0,i_0)a_0 \dots a_{n-1} (q_n,i_n)$
	is an accepting 
	run of $\delta_k(\calN)$, and for every $0 \leq j < n$, 
	\begin{align*}
		& w'((q_{j},i_{j}),a_j,(q_{j+1},i_{j+1})) \\
		{ }={ } & i_{j} + w(q_{j},a_j,q_{j+1}) - i_{j+1}
		= w''(p_{j},a_j,p_{j+1}).
	\end{align*}
	Therefore $\calM \subseteq_0 \delta_k(\calN)$.

	\item \ref{lem-item:del_equ})
	This property follows immediately from the first property, the second
	property in the particular case $\calM = \calN$, and Lemma
	\ref{lem:k-inc_prop} item \ref{lem-item:k-inc_prop_sub}.
\end{proof}

The next result follows immediately from the preceding Lemma and
Lemma~\ref{lem:k-inc_prop} item \ref{lem-item:k-inc_prop_trans}.

\begin{proposition}\label{pro:k-to-zero-delay}
	An automaton $\calN$ is $k$-delay determinizable if and only if
	$\delta_k(\calN)$ is $0$-delay determinizable.
\end{proposition}

The above result raises the question of whether, for all $k$, $0$-delay
determinization can be reduced to $k$-delay determinization. We give a positive
answer to this question in the form of Lemma~\ref{lem:0-to-k} in
Section~\ref{sec:lower-bound-delay}. 

We now proceed with the proof of Theorem~\ref{thm:kdelay-in-exp}.

\begin{proof}[Proof of Theorem~\ref{thm:kdelay-in-exp}]
	It should be clear that \TWOEXP~membership follows from
	Proposition~\ref{pro:k-to-zero-delay} and
	Proposition~\ref{pro:0delay-in-exp}. We now observe that the subset
	construction used to decide $0$-delay determinizability need only be
	applied on the first component of the state space resulting from the use
	of the delay construction. In other words, once both constructions are
	applied, a state will correspond to a function $f : Q \to \{-k, \dots,
	k\} \cup \{\bot\}$, where $q \mapsto \bot$ signifies that $q$ is not in
	the subset. The size of the resulting state space is then
	$2^{\mathcal{O}(|Q| \log_2k)}$.  Thus, the composition of this two
	constructions yields only a single exponential.
\end{proof}

\subsection{Lower bound}\label{sec:lower-bound-delay}
We reduce the $0$-delay uniformization problem for synchronous transducers to
that of deciding whether a given automaton is $k$-delay determinizable (for any
fixed $k \in \mathbb{N}$). As the former problem is known to be \EXP-complete
(see~\cite{fjlw16}), this implies the latter is \EXP-hard.

\begin{theorem}\label{thm:kdelay-exp-hard}
	Deciding the $k$-delay problem for a given automaton is
	\EXP-hard, even for fixed $k \in \mathbb{N}$.
\end{theorem}

For convenience, we will first prove that the $0$-delay problem reduces to the
$k$-delay problem for any fixed $k$. We then show the former is \EXP-hard.
\begin{lemma}\label{lem:0-to-k}
	The $0$-delay problem reduces in logarithmic space to the $k$-delay
	problem, for any fixed $k \in \mathbb{N}$.
\end{lemma}

Let us fix some $k \in \mathbb{N}$. Given the automaton $\calN =
(Q,I,A,\Delta,w,F)$ we denote by $x \cdot \calN$ the automaton $(Q,I,A,\Delta,x
\cdot w,F)$, where $x \cdot w$ is such that $d \mapsto x \cdot w(d)$ for all $d
\in \Delta$. Lemma~\ref{lem:0-to-k} is a direct consequence of the following.
\begin{lemma}
	For every automaton $\calN = (Q,I,A,\Delta,w,F)$, the following statements are equivalent.
	\begin{enumerate}[noitemsep, nolistsep]
	\item \label{N1}
	$\calN$ is $0$-delay determinizable;
	\item \label{N2}
	$(4k + 1) \cdot \calN$ is $0$-delay determinizable;
	\item \label{N3}
	$(4k + 1) \cdot \calN$ is $k$-delay determinizable.
	\end{enumerate}
\end{lemma}
\begin{proof}
	Given a $0$-delay
	determinizer $\calD$ of $\calN$, the automaton $(4k + 1) \cdot \calD$ is easily seen to be a
	$0$-delay determinizer of $(4k + 1) \cdot \calN$.
	This proves that the first statement implies the second one.
	Moreover, as a direct consequence of Lemma \ref{lem:k-inc_prop} item \ref{lem-item:k-inc_prop_less}, the second
	statement implies the third one.
	To complete the proof, we argue that if $(4k + 1) \cdot \calN$ is $k$-delay
	determinizable, then $\calN$ is $0$-delay determinizable.  Let $\calD' =
	(Q',I',A,\Delta',w',F')$ be a $k$-delay determinizer of $(4k + 1) \cdot
	\calN$. Let $\gamma$ be the function mapping every integer $x$ to the 
	unique integer $\gamma(x)$ satisfying $|(4k+1)\gamma(x) - x| \leq 2k$, and let $\gamma(\calD')$ denote the
	deterministic automaton $(Q',I',A,\Delta',\gamma \circ w', F')$. 
	We now argue that $\gamma(\calD')$ is a $0$-delay determinizer of $\calN$.
	For every sequence $\rho' = q'_0 a_0 \dots a_{n-1} q'_n \in (Q' \cdot A)^* Q'$, since the states and transitions of $\calD'$ and $\gamma(\calD')$ are identical, $\rho'$
	is an accepting run of $\calD'$ if and only if it is an accepting run of $\gamma(\calD')$.
	Therefore, since $\calD'$ is a $k$-delay determinizer of $(4k + 1) \cdot \calN$, if $\rho'$ is an accepting run of
	$\gamma(\calD')$, there exists an accepting run $\rho = q_0 a_0 \dots a_{n-1} q_n$ of $\calN$ such that $(4k + 1)w(\rho) =
	w'(\rho')$, and for every $0 \leq i \leq n$, $|(4k + 1)w(q_0 \dots q_i) - w'(q_0' \dots q_i')| \leq k$.
	As a consequence, for every $1 \leq i \leq n$ we have
	\[
	\begin{array}{lll}
	& & |(4k + 1)w(q_{i-1} a_{i-1} q_i) - w'(q_{i-1}' a_{i-1} q_{i}')|\\
	& = & |(4k + 1)w(q_0 \dots q_i) - (4k + 1)w(q_0 \dots q_{i-1}) - w'(q_0' \dots q_i') + w'(q_0' \dots q_{i-1}')|\\
	& \leq & |(4k + 1)w(q_0 \dots q_i) - w'(q_0' \dots q_i')| + |w'(q_0' \dots q_{i-1}') - (4k + 1)w(q_0 \dots q_{i-1})| \leq 2k,
	\end{array}
	\]
	hence $w(q_{i-1} a_{i-1} q_i) = \gamma(w(q_{i-1}' a_{i-1} q_{i}'))$. As a consequence, $\gamma(\calD')$ is a $0$-delay determinizer of $\calN$.
\end{proof}

We can now show that the $k$-delay problem is \EXP-hard by arguing that the
$0$-delay problem is \EXP-hard. Let us introduce some notation
regarding transducers.

\myparagraph{Transducers}
A \emph{(synchronous) transducer} $\calT$ from an input alphabet $A_I$ to an
output alphabet $A_O$ is an unweighted automaton $(Q,I,A_I \times
A_O,\Delta,F)$. We denote the \emph{domain} of $\calT$ by $\dom(\calT) \defeq \{
a_0 \dots a_{n-1} \in A_I^* \st (a_0,b_0) \dots (a_{n-1},b_{n-1}) \in
(A_I \times A_O)^*\}$. The transducer $\calT$ is said to be
\emph{input-deterministic} if
for all $p \in Q$, for all $a \in A_I$, there
exist at most one state-output pair $(q,b) \in Q \times A_O$ such that
$(p,(a,b),q) \in \Delta$.
A transducer $\calU$ from $A_I$ to $A_O$ is a
\emph{$0$-delay uniformizer} of $\calT$ if
\begin{inparaenum}[$(i)$]
	\item $\calU$ is input-deterministic,
	\item $\lang{\calU} \subseteq \lang{\calT}$, and
	\item $\dom(\calU) = \dom(\calT)$.
\end{inparaenum}
If such a transducer exists, we say $\calT$ is $0$-delay uniformizable. Given a
transducer, to determine whether it is $0$-delay uniformizable is an \EXP-hard
problem~\cite{fjlw16}.

Intuitively, a transducer induces a relation from input words to output words.
We construct an automaton that replaces the output alphabet by unique positive
integer identifiers. For convenience, we also make sure the constructed
automaton defines a function which maps every word in its language to $0$. 

\myparagraph{From transducers to weighted automata}
Given a transducer $\calT = (Q,I,A_I \times A_O,\Delta,F)$ with $A_0 = \{1,
\dots M\}$, we construct a weighted automaton $\calN_\calT =
(Q',I,A_I\cup\{\#\},\Delta',w,F)$ as follows:
\begin{itemize}[nolistsep]
	\item $Q' = Q \cup Q \times A_O$,
	\item $\Delta' = \{ \left( p, a, (q,m) \right),
			    \left( (q,m), \#, q \right) \st 
		(p,(a,m),q) \in \Delta\}$,
	\item $w : \Delta' \to \mathbb{Z}$, $\left( p, a,
		(q,m) \right) \mapsto m$ and $\left( (q,m), \#,
		q\right) \mapsto -m$.
\end{itemize}
\begin{lemma}\label{lem:props-trans-wa}
	The translation from transducers to weighted automata satisfies the
	following properties.
	\begin{enumerate}[nolistsep]
		\item $q_0 (a_0,m_0) \dots (a_{n-1},m_{n-1}) q_n$ is a run of
			$\calT$ if and only if $q_0 a_0 (q_0, m_0) \# \dots
			(q_{n-1},m_{n-1}) \# q_n$ is a run of $\calN_\calT$.
			Moreover, for all $0 \le i \le n$
			\begin{itemize}[nolistsep]
				\item $w(q_0 \dots q_i) = 0$, and
				\item $w(q_0 \dots (q_i,m_i)) = m_i$;
			\end{itemize}
		\item $\lang{\calN_\calT} = \{a_0 \# \dots \# a_n \# \st a_0
			\dots a_n \in \dom(\calT)\}$;
		\item $\inter{\calN_\calT}(\alpha) = 0$ for all $\alpha
			\in \lang{\calN_\calT}$.
	\end{enumerate}
\end{lemma}
\begin{proof}
	The first item follows by construction of the automaton $\calN_\calT$.
	Items $2$ and $3$ are direct consequences of item $1$.
\end{proof}

We are now ready to show the $0$-delay uniformization problem reduces in
polynomial time to the $0$-delay determinization problem. To do so, we show
that any $0$-delay uniformizer of a transducer $\calT$ can be
transformed into a $0$-delay determinizer of $\calN_\calT$,
and vice versa.

\begin{lemma}\label{lem:0delay-exp-hard}
	Deciding the $0$-delay problem for a given automaton is
	\EXP-hard.
\end{lemma}
\begin{proof}
	Given a transducer $\calT = (Q,I,A_I \times A_O,\Delta,F)$ with $A_0 =
	\{1,\dots,M\}$, we construct
	$\calN_\calT = (Q',I,A_I\cup\{\#\},\Delta',w,F)$.
%
	Suppose $\calU = (S,\{s_0\},A_I \times A_O,R,G)$ is a $0$-delay
	uniformizer of $\calT$. 
	%
	%
	In what follows we argue that
	$\calN_\calU$ is 
	a $0$-delay determinizer of $\calN_{\calT}$. 
	Since $\calU$ is input-deterministic, the automaton $\calN_\calU$ is
	deterministic. Also, since $\calU$ is a $0$-delay uniformizer of
	$\calT$, then we have that $\dom(\calU) = \dom(\calT)$. Hence, from
	Lemma~\ref{lem:props-trans-wa} item $2$ it follows that $\lang{\calN_\calT}
	= \lang{\calN_\calU}$.
	Since both automata map their languages to the value $0$ (see
	Lemma~\ref{lem:props-trans-wa} item $3$), we have that
	$\inter{\calN_\calU} = \inter{\calN_\calT}$. Finally, by using
	Lemma~\ref{lem:props-trans-wa} item $1$, we get that $\calN_\calU
	\subseteq_0 \calN_\calT$ from the fact that $\lang{\calU} \subseteq
	\lang{\calT}$.

	Assume $\calD = (S,\{s_0\},A_I \cup\{\#\},R,\mu,G)$ is a $0$-delay
	determinizer of $\calN_\calT$.
	Let $\calU$ be the transducer $(S,\{s_0\},A_I
	\times A_O, R',G)$ where $R' = \{ \left( p,(a,m),s )\right) \st
	(p,a,q), (q,\#,s) \in R \land \mu(p,a,q) = -\mu(q,\#,s) = m \}$.
	Since $\calD$ is deterministic, we have
	that $\calU$ is input-deterministic. By construction, we have that
	\[
		\dom(\calU) = \{ a_0 a_1 \dots a_n \st a_0 \# a_1 \# \dots \#
			a_n \# \in \lang{\calD}\}.
	\]
	Therefore, since $\lang{\calN_\calT} = \lang{\calD}$, from
	Lemma~\ref{lem:props-trans-wa} item $2$ we get that $\dom(\calU) =
	\dom(\calT)$. Also, by construction, we have that
	\(
		s_0(a_0,m_0) \dots (a_{n-1},m_{n-1}) s_n
	\)
	is a run of $\calU$ if and only if $s_0 a_0 q_0 \# s_1 a_1 q_1 \# \dots
	q_{n-1}	\# s_n$ is a run of $\calD$ such that $\mu(s_0 \dots q_i) = m_i$ for all $0 \le i \le n$.
	Moreover, since $\calD$ is a $0$-delay determinizer of $\calN_\calT$, $\mu(s_0 \dots s_i) = 0$ for all $0 \le i \le n$ (see Lemma~\ref{lem:props-trans-wa} item $3$).
	Finally, because $\calD \subseteq_0 \calN_\calT$, we get that
	$\lang{\calU} \subseteq \lang{\calT}$ by
	Lemma~\ref{lem:props-trans-wa} item $1$ and the above argument. 
\end{proof}

\section{Deciding $r$-regret determinizability}
In this section we argue that the $r$-regret problem is \EXP-complete. It will
be convenient to \textbf{suppose all automata we work with are trim}.
This is no loss of generality with regard to
$r$-regret determinizability, \ie~an automaton $\calN$ is $r$-regret
determinizable if and only if its trim version, $\calN'$, is $r$-regret
determinizable. Clearly, an $r$-regret determinizer of $\calN'$ is also an
$r$-regret determinizer of $\calN$. Also, it is easy to show that the trim
version $\calD'$ of an $r$-regret determinizer $\calD$ of $\calN$ must also be
an $r$-regret determinizer of $\calN'$.  Furthermore, any automaton can be
trimmed in polynomial time.
%

\subsection{Upper bound}\label{sec:upper-regret}
We will now give an exponential time algorithm to determine whether a given
automaton is $r$-regret determinizable, for a given $r$. The algorithm is based
on a quantitative version of the \emph{Joker game} introduced by Kuperberg and
Skrzypczak to study the determinization of good-for-games automata~\cite{ks15}.
More precisely, the Joker game will correspond to generalization of the
classical energy games~\cite{cdahs03}.

The algorithm is as follows: construct an energy game with resets (which we call
the Joker game) based on the given automaton and decide if \eve wins it; if this
is not the case, then for all $r \in \mathbb{N}$ the automaton is not $r$-regret
determinizable; otherwise, using the winning strategy for \eve in the Joker
game, construct a deterministic automaton $\calD$ realizing the same function as
$\calN$ and use it to decide if $\calN$ is $r$-regret determinizable. The last
step of the algorithm is the simplest. Given a deterministic version of the
original automaton, one can use it as a ``monitor'' and reduce the $r$-regret
determinizability problem to deciding the winner in an energy game.
\begin{theorem}\label{thm:regret-upper}
	Deciding the $r$-regret problem for a given automaton is in
	\EXP.
\end{theorem}

\myparagraph{Energy games with resets}
An \emph{energy game with resets} (EGR for short) is an infinite-duration
two-player turn-based game played by \eve and \adam on a directed weighted
graph. Formally, an EGR $\calG = (V,V_\exists, E, E_\rho,w)$ consists of: a set
$V$ of vertices, a set $V_\exists \subseteq V$ of vertices of \eve---the set
$V_\forall \defeq
(V \setminus V_\exists)$ of vertices thus belongs to \adam, a set $E \subseteq
V \times V$ of directed edges, a set $E_\rho \subseteq E$ of \emph{reset edges}
such that $E_\rho \subseteq V_\forall \times V$, and a weight
function $w : E \to \mathbb{Z}$.
(Observe that if $E_\rho = \emptyset$, we obtain the classical energy
games without resets~\cite{cdahs03}.)
Pictorially, we represent \eve vertices by
squares and \adam vertices by circles. We denote by $\maxweight$ the value
$\max_{e \in E} |w(e)|$.  Intuitively, from the current vertex $u$, the player
who owns $u$ (\ie~\eve if $u \in V_\exists$, and \adam otherwise)
chooses an edge $(u,v) \in
E$ and the \emph{play} moves to $v$. We formalize the notions of strategy and
play below. 

\begin{figure}
\begin{center}
\begin{tikzpicture}
	\node[ve] (v0) {$v_0$};
	\node[va,right=1.8cm of v0] (v1) {$v_1$};
	\node[ve,above right=1cm of v0] (v2) {$v_2$};

	\path
	(v0) edge[bend left] node[el] {$1$} (v1)
	(v1) edge[bend left] node[el,swap] {$-1$} (v0)
	(v1) edge node[el,swap] {$0$} (v2)
	(v2) edge node[el,swap] {$-2$} (v0)
	;
\end{tikzpicture}
\caption{Energy game with reset edges $E_\rho = \{(v_1,v_2)\}$ where \eve wins
from $v_0$ with initial credit $c_0 = 3$}
\label{fig:egr-example}
\end{center}
\end{figure}
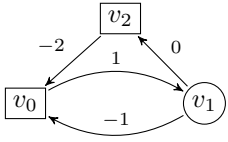

A strategy for \eve (respectively, \adam) in $\calG$ is a mapping $\sigma :
V^* \cdot V_\exists \to V$ (respectively, $\tau : V^* \cdot V_\forall \to V$)
such that $\sigma(v_0 \dots v_n) = v_{n+1}$ ($\tau(v_0 \dots v_n) = v_{n+1}$)
implies $(v_n, v_{n+1}) \in E$. As in regret games, a strategy $\sigma$ for
either player is one which can be encoded as a deterministic Mealy machine
$(S,s_I,\lambda_u,\lambda_o)$ with update function $\lambda_u : S \times V \to S$
and output function $\lambda_o : S \times V \to V$. The machine encodes $\sigma$
in the following sense: $\sigma(v_0 \dots v_n) = \lambda_u(s_n,v_{n})$ where $s_0
= s_I$ and $s_{i+1} = \lambda_u(s_i,v_i)$ for all $0 \le i < n$. As usual,
the memory of a finite-memory strategy refers to the size of the
Mealy machine realizing it.

A play in $\calG$ from $v \in V$ corresponds to an infinite
path in the underlying directed graph $(V,E)$. That is, a sequence $\pi = v_0
v_1 \dots$ such that $(v_i,v_{i+1}) \in E$ for all $i \in \mathbb{N}$. Since an
EGR is played for an infinite duration,
we will henceforth assume they are
played on digraphs with no sinks: \ie for all $u \in V$, there exists $v \in V$
such that $(u,v) \in E$.
We say a play $\pi = v_0 v_1 \dots$ is consistent with a strategy $\sigma$ for
\eve (respectively, $\tau$ for \adam) if it holds that $v_i \in V_\exists$
implies $\sigma(v_0 \dots v_i) = v_{i+1}$ ($v_i \not\in V_\forall$ implies
$\tau(v_0 \dots v_i) = v_{i+1}$).  Given a strategy $\sigma$ for \eve and a
strategy $\tau$ for \adam, and a vertex $v \in V$ there is a unique play
$\pi^v_{\sigma \tau}$ compatible with both $\sigma$ and $\tau$ from $v$.

Given a \emph{finite path} $\phi$ in $\calG$, \ie a sequence $v_0 \dots v_n$
such that $(v_i,v_{i+1}) \in E$ for all $0 \le i < n$, and an initial credit
$c_0 \in \mathbb{N}$, we define the \emph{energy level} of $\phi$ as
$\el_{c_0}(\phi) \defeq c_0 + \sum^{n-1}_{j = i_0} w(v_j,v_{j+1})$ where $0 \le
i_0 < n$ is the minimal index such that $(v_{\ell},v_{\ell + 1}) \not\in E_\rho$
for all $i_0 < \ell < n$.

We say \eve wins the EGR from a vertex $v \in V$ with initial credit $c_0$ if
she has a strategy $\sigma$ such that, for all strategies $\tau$ for \adam, for
all finite prefixes $\phi$ of $\pi^v_{\sigma \tau}$ we have $\el_{c_0}(\phi) \ge
0$. \adam wins the EGR from $v$ with initial credit $c_0$ 
if and only if \eve does not win it.

\begin{example}
	Consider the EGR shown in Fig.~\ref{fig:egr-example}. In this game,
	\eve wins from $v_0$ with initial credit $2$. Indeed, whenever
	\adam plays from $v_1$ to $v_0$ the energy level drops by $1$ but is
	then increased by $1$ when the play returns to $v_1$; when he
	plays from $v_1$ to $v_2$ the energy level is first reset to $2$ and
	then drops to $0$ when the play reaches $v_0$. Clearly then, \adam
	cannot force a negative energy level. However, if $E_\rho$ were empty,
	then \eve would lose the game regardless of the initial credit.
\end{example}

The following properties of energy games (both, with or without resets), which
include \emph{positional determinacy}, will be useful in the sequel. A game is
positionally determined if: for all instances of the game, from all
vertices, either \eve has a positional strategy which is winning for her against
any strategy for \adam, or \adam has a positional strategy which is winning for
him against any strategy for \eve.

\begin{proposition}\label{pro:eg-props}
	For any energy game (both, with or without resets)
	$\calG = (V,V_\exists,E,E_\rho,w)$ the following hold.
	\begin{enumerate}[nolistsep]
		\item\label{itm:pos-det} The game is positionally
			determined if $c_0 \ge |V|\maxweight$.
		\item\label{itm:any-ic} For all $v \in V$, \eve wins from $v \in
			V$ with initial credit $|V|\maxweight$ if and only if
			there exists $c_0 \in \mathbb{N}$ such that she wins
			from $v \in V$ with initial credit $c_0$.
		\item Determining if there exists $c_0 \in \mathbb{N}$ such that
			\eve wins from $v \in V$ with initial credit $c_0$
			is decidable in time polynomial in $|V|$, $|E|$, and
			$\maxweight$.
	\end{enumerate}
\end{proposition}
\begin{proof}[Sketch]
	All three properties are known to hold for energy games without
	resets (see, \eg~\cite{cdahs03,bcdgr09}).

	For EGRs the arguments to show these properties hold are
	almost identical to those used in~\cite{bcdgr09}. We first define a
	finite version of the game which is stopped after the first cycle is
	formed and in which the winner is determined based on properties of that
	cycle. If we let \eve win if and only if the cycle has non-negative sum
	of weights or it contains a reset, then we can show she wins this
	\emph{First Cycle Game}~\cite{ar14} if and only if she wins the EGR with
	initial credit $|V|\maxweight$. Furthermore, using a result
	from~\cite{ar14} we obtain that positional strategies suffice for both
	players in both games, \ie~the games are positionally determined.
	
	The second property follows from the relationship between the EGR and the
	first cycle game we construct. More precisely, we show that winning
	strategies for both players transfer between the games. In the first
	cycle game, \adam wins if he can force cycles which have a negative sum
	of weights. Hence, if \eve does not win the EGR with initial credit
	$|V|\maxweight$, then by determinacy, \adam wins the first cycle game,
	and his strategy---when played on the original EGR---ensures only
	negatively-weighted cycles are formed, which in turn means that he wins
	the EGR with any initial credit.
	
	Finally, to obtain an algorithm, we reduce the problem of deciding if
	\eve wins the EGR from $v \in V$ with a given initial credit $c_0$ to
	her winning a \emph{safety game}~\cite{ag11} played on an unweighted
	digraph where the states keep track of the energy level (up to a maximum
	of $|V|\maxweight$).
\end{proof}

Energy games will be our main tool for the rest of this section. They allow us
to claim that, given an automaton $\calN$ and a deterministic automaton $\calD$
which defines the same function, we can decide $r$-regret determinizability.
\begin{proposition}\label{pro:r-regret-eg}
	Given an automaton $\calN = (Q,I,A,\Delta,w,F)$ and $\calD =
	(Q',\{q'_I\},A,\Delta',w',F')$ such that $\calD$ is deterministic and
	$\inter{\calD} = \inter{\calN}$, the $r$-regret problem for
	$\calN$ is decidable in time polynomial in $|Q|$, $|Q'|$,  $|A|$, 
	$\maxweight$, and $\maxweight'$.
\end{proposition}
\begin{proof}[Sketch]
	We construct an energy game without resets
	which simulates the regret game played on $\calN$
	while using $\calD$ to compare the weights of transitions chosen by \eve
	to those of the maximal run of $\calN$. Intuitively, \eve chooses an
	initial state in $\calN$, then \adam chooses a symbol, and \eve responds
	with a transition $t \in \Delta$ in $\calN$. Finally, the state of
	$\calD$ is deterministically updated via transition $t'$. The weight of
	the whole round is $w(t) - w'(t')$. We also make sure \eve
	loses if in the regret game she reaches a non-final state and the run of
	$\calD$ is at a final state; or if she reaches a final state with too
	low an energy level (implying a large regret).
	Formally, the energy
	game without resets is $\calG = (V,V_\exists,E, \emptyset,\mu)$ where:
	\begin{itemize}[nolistsep]
		\item $V = Q^2 \cup Q^3 \times A \cup
			\{\top,\bot,\bot^r_1,\bot^r_2\}$;
		\item $V_\exists =  Q^3 \times A$;
		\item $E$ contains edges to simulate transitions of $\calN$
			and $\calD$, \ie 
			$\{ \left((p,q),(p,q,q',a)\right) \st (q,a,q') \in
			\Delta' \} \cup \{ \left((p,q,q',a),(p',q')\right) \st
			(p,a,p') \in \Delta \}$, edges required
			to verify \eve does not reach a non-final state when
			$\calD$ accepts, \ie~$\{ \left((p,q),\bot\right) \st p
			\not\in F \land q \in F' \} \cup \{ (\bot,\bot) \}$,
			edges used to make sure the regret is at most $r$ when
			on final states,
			\ie~$\{ \left((p,q),\bot^r_1\right) \st p \in F \land q
			\in F'\} \cup \{(\bot^r_1,\bot^r_2),
			(\bot^r_2,\bot^r_1)\}$, and edges to punish one of the
			players if an automaton blocks, \ie~$\{ \left( (p,q), \top
			\right) \st \lnot \exists (q,a,q') \in \Delta'\} \cup
			\{ \left( (p,q,q',a), \bot \right) \st \lnot \exists
			(p,a,p') \in \Delta\} \cup \{(\top,\top)\}$;
		\item $\mu : E \to \mathbb{Z}$ is such that
			\begin{itemize}[nolistsep]
				\item $\left( (p,q,q',a),(p',q') \right) \mapsto
					w(p,a,p') - w'(q,a,q')$,
				\item $(\bot,\bot) \mapsto -1$,
				\item $\left( (p,q),\bot^r_1 \right) \mapsto 1 -
					|Q'|(\maxweight + \maxweight')$,
				\item $(\bot^r_1,\bot^r_2) \mapsto -1$,
					$(\bot^r_2,\bot^r_1) \mapsto 1$,
				\item $(\top,\top) \mapsto 1$, and
				\item $e \mapsto 0$ for all other $e \in E$.
			\end{itemize}
	\end{itemize}
	We then claim that for some $p_I \in I$, \eve wins the energy game
	without resets $\calG$ from $(p_I,q'_I)$ with initial credit $r +
	|Q'|(\maxweight + \maxweight')$ if and only if $\calN$ is $r$-regret
	determinizable. The result then follows from the fact $\calG$ is of size
	polynomial w.r.t. $\calD$ and $\calN$, and the application of the
	algorithm (see Proposition~\ref{pro:eg-props}) to
	determine the winner of $\calG$.
\end{proof}

\myparagraph{The Joker game}
The Joker game (JG) is a game played by \eve and \adam on an automaton 
$(Q,I,A,\Delta,w,F)$. It is played as follows:
\eve chooses as initial state $p \in I$ and \adam an initial state $q \in I$ and
the initial configuration becomes $(p,q) \in I^2$.  From the current
configuration $(p,q) \in Q^2$
\begin{inparaenum}
\item[\textbf{(Step i):}] \adam chooses a symbol $a \in A$,
\item[\textbf{(Step ii):}] then \eve chooses a transition $(p,a,p') \in \Delta$,
	and
\item[\textbf{(Step iii):}] \adam can
	\begin{inparaenum}
	\item[\textbf{(Step iii.a):}] choose a transition $(q,a,q') \in
		\Delta$ or
	\item[\textbf{(Step iii.b):}] play $\joker$ and choose a
		transition $(p,a,q') \in \Delta$.
	\end{inparaenum}
\end{inparaenum}
The new configuration is then $(p',q')$.  The weight assigned to each round
corresponds to the weight of the transition chosen by \eve minus the weight of
that chosen by \adam.  If \adam played $\joker$, then the sum of weights is
reset before adding the weight of the configuration change. Additionally, if
\eve moves to a non-final state and \adam moves to a final state, or if 
\eve can no longer extend the run she is constructing, then
\textbf{(Step $\star$):} we ensure \eve loses the game.

We formalize the JG played on $\calN = (Q,I,A,\Delta,w,F)$ as an EGR
$(V,V_\exists,E,E_\rho,\mu)$ with $V = V_\exists \cup V_\forall$, $E =
\bigcup_{1 \le i \le 3} E_{\forall_i} \cup E_\exists \cup E_\rho \cup
\{(\bot,\bot)\}$ where:
\begin{itemize}[nolistsep]
	\item $V_\exists = Q^2 \times A \cup \{\bot\}$;
	\item $V_\forall = Q^2 \cup Q^3 \times A$;
	\item \textbf{(Step i):} $E_{\forall_1} = \{ \left( (p,q), (p,q,a)
		\right) \st (p,q) \in Q^2, (q,a,q') \in \Delta\}$;
	\item \textbf{(Step ii):} $E_\exists = \{ \left( (p,q,a), (p,q,p',a)
		\right) \st (p,q,a) \in V_\exists, (p,a,p') \in \Delta \}$;
	\item \textbf{(Step iii.a):} $E_{\forall_2} = \{ \left( (p,q,p',a),
		(p',q') \right) \st (p,q,p',a) \in Q^3 \times A, (q,a,q') \in
		\Delta \}$;
	\item \textbf{(Step iii.b):} $E_\rho = \{ \left( (p,q,p',a), (p',p'')
		\right) \st (p,q,p',a) \in Q^3 \times A, (p,a,p'') \in \Delta
		\}$;
	\item \textbf{(Step $\star$):} $E_{\forall_3} = \{ \left( (p,q), \bot
		\right) \st p \not\in F \land q \in F \text{ or } \exists a \in
		A, \forall p' \in Q : (p,a,p') \not\in \Delta \}
		\cup \{(\bot,\bot)\}$; and
	\item $\mu$ is such that
		\begin{itemize}[nolistsep]
			\item $(\bot,\bot) \mapsto -1$,
			\item $e \mapsto w(p,a,p') - w(q,a,q')$ for all $e =
				\left( (p,q,p',a), (p',q') \right) \in
				E_{\forall_2}$,
			\item $e \mapsto w(p,a,p') - w(p,a,p'')$ for all $e
				= \left( (p,q,p',a) ,(p,p'') \right) \in
				E_{\rho}$,
			\item and $e \mapsto 0$ for all other $e \in E$.
		\end{itemize}
\end{itemize}
It is easy to verify that there are no sinks in the EGR.

\myparagraph{Winning the Joker game}
We say \eve wins the JG played on $(Q,I,A,\Delta,w,F)$ if there is $p \in I$ such
that, for all $q \in I$, she wins from $(p,q)$ with initial credit
$|V|\maxm$ (where $\maxm \defeq \max_{e \in E} |\mu(e)|$).
Proposition~\ref{pro:eg-props} tells us that, if \eve wins with some initial
credit, then she also wins with initial credit $|V|\maxm$.

We now establish a relationship between $r$-regret determinization and the JG.
\begin{lemma}\label{lem:joker-nec-regret}
	If an automaton $\calN = (Q,I,A,\Delta,w,F)$ is $r$-regret
	determinizable, for some $r \in \mathbb{N}$, then \eve wins the JG
	played on $\calN$.
\end{lemma}
\begin{proof}
	We will actually prove the contrapositive holds.  Suppose \eve does not
	win the JG. By determinacy of EGRs (Proposition~\ref{pro:eg-props}
	item~\ref{itm:pos-det}) we know that \adam, for all $p_0 \in I$, has a
	strategy $\tau$ to force from some $(p_0,q_0) \in I^2$ a play
	which eventually witnesses a negative energy level.  Furthermore, he can
	do so for any initial credit
	(Proposition~\ref{pro:eg-props} item~\ref{itm:any-ic}). Let us now
	assume, towards a contradiction, that \eve wins the $r$-regret game with
	a strategy $\sigma$ such that $\sigma(\epsilon) = p_0$.
	Since $\sigma$ is winning for her in the regret game,
	then for all $\alpha \in A^*$, $\sigma(\alpha)$ is an initial run of
	$\calN$. Hence $\sigma$ can be converted into a strategy for \eve in the
	JG by ignoring the transitions chosen by \adam and following $\sigma$
	when \adam chooses a symbol $a \in A$.
	If \eve follows $\sigma$ to play
	in the JG against $\tau$, then there exists $q_0 \in I$ such that
	$\pi^{(p_0,q_0)}_{\sigma \tau}$ eventually witnesses a
	negative energy level even if the initial credit is $r + 2|Q|\maxweight$
	(because $\tau$ is winning for \adam in the JG with any initial credit).
	Moreover, $\pi^{(p_0,q_0)}_{\sigma \tau}$ never reaches the vertex
	$\bot$, since $\sigma(\alpha)$ is an initial run of $\calN$ for all
	$\alpha \in A^*$.
	If we let $\phi = (p_0,q_0) (p_0,q_0,a_0) (p_0,q_0,p_1,a_0) \dots
	(p_n,q_n)$ be the first prefix that witnesses a negative energy level
	with initial credit $r + 2|Q|\maxweight$, and $0 \le i_0 < n$ be the
	minimal index such that no reset occurs for all $i_0 < \ell < n$, then
	$\rho = p_0 a_0 p_1 a_1 \dots p_n$ and $\rho' = p_0 a_0 \dots p_{i_0}
	a_{i_0} q_{i_0 + 1} \dots q_n$ are two runs in $\calN$ such that
	$w(\rho') > w(\rho) + r + 2|Q|\maxweight$. Since
	$\calN$ is trim, there is a final run $q_n a_n \dots a_{m-1} q_m$
	such that $m - n \le |Q|$.  Hence, we have that $\inter{\calN}(a_0 \dots
	a_{m-1}) - \Val(\sigma(a_0 \dots a_{m-1}))  > r$, which contradicts the
	fact that $\sigma$ is winning for \eve in the regret game. It follows
	that there cannot be a winning strategy for \eve in the $r$-regret game.
\end{proof}

From the above results we have that if we construct the JG for the given
automaton and \eve does not win the JG, then the automaton cannot be $r$-regret
determinizable (no matter the value of $r$). We now study the
case when \eve does win.

\myparagraph{Using the JG to determinize an automaton}
Let $\calN = (Q,I,A,\Delta,w,F)$ be an automaton.  We will assume that \eve wins
the JG played on $\calN$. Denote by $W^{\JG} \subseteq Q^2$ the
\emph{winning region} of \eve. That is, $W^{\JG}$ is the set of all $(p,q) \in
Q^2$ such that \eve wins the EGR from $(p,q)$ with initial credit
$|V|\maxm$.  Also, let us write $Q^{\JG}$ for the projection of $W^{\JG}$
on its first component.  Moreover, for every $(p,q) \in W^{\JG}$, let
$\credit(p,q)$ denote the minimal integer $c \in \mathbb{N}$ such that \eve wins
the JG from $(p,q)$ with initial credit $c$.

We will now prove some properties of the sets $W^{\JG}$ and $Q^{\JG}$.
First, the relation $W^{\JG}$ is
transitive.
\begin{lemma}\label{lem:transitivity}
	For all
	$p,q,t \in Q$, if $(p,q),(q,t) \in W^{\JG}$ then $(p,t) \in W^{\JG}$.
\end{lemma}
\begin{proof}[Sketch]
	For every $(p,q) \in W^{\JG}$, let $\sigma_{(p,q)}$ denote a winning
	strategy for \eve in the JG played from $(p,q)$ with initial credit
	$\credit(p,q)$. We define a strategy $\sigma$ for \eve in the JG as
	follows. For $(p,q),(q,t) \in W^{\JG}$, let $q_{p,t} \in Q$ denote the
	state such that $\credit(p,q_{p,t}) + \credit(q_{p,t},t)$ is minimal.
	For every $(p,t,a) \in Q^2 \times A$, if $(p,q),(p,t) \in W^{\JG}$, we
	then set $\sigma\left((p,t,a)\right) =
	\sigma_{(p,q_{p,t})}((p,q_{p,t})(p,q_{p,t},a))$.  We then claim that for
	every $(p,q)(q,t) \in W^{\JG}$, the strategy $\sigma$ is winning for
	\eve in the EGR starting from $(p,t)$ with initial credit
	$\credit(p,q_{p,t}) + \credit(q_{p,t},t)$.
\end{proof}

Another property which will be useful in the sequel is that, all the
$a$-successors of a state $p \in Q^{\JG}$ are related (by $W^{\JG}$) to the
$a$-successor chosen by a winning strategy for \eve.
\begin{lemma}\label{lem:joking}
	For all
	$(p,q) \in W^{\JG}$ and $a \in A$, let $\sigma^{\JG}$ be a winning
	strategy for \eve in the JG from $(p,q)$ with initial credit $c \in
	\mathbb{N}$, and let $(p,q,p',a) = \sigma^{\JG}\left((p,q) (p,q,a)\right)$.
	Then, for all $(t,a,p'') \in \Delta$ such that $t \in \{p,q\}$, it holds
	that $(p',p'') \in W^{\JG}$, and $\credit(p',p'') \leq c + w(p,a,p') -
	w(t,a,p'')$.
\end{lemma}
\begin{proof}
	Observe that from any $(p,q) \in W^{\JG}$ in the JG, after \adam has
	chosen a letter $a \in A$ and \eve a transition $(p,a,p') \in \Delta$,
	\adam could play $\joker$ and choose any transition $(p,a,p'') \in
	\Delta$ or (without playing $\joker$) choose any transition $(q,a,p'')$.
	Hence, for any winning strategy $\sigma^{\JG}$ for \eve in the JG played
	from $(p,q)$ with initial credit $c$ such that $\sigma^{\JG}\left(
	(p,q)(p,q,a) \right) = (p,q,p',a)$, for any $(t,a,p'') \in \Delta$ such
	that $t \in \{p,q\}$, reaching $(p',p'')$ is consistent with
	$\sigma^{\JG}$. It follows that $\sigma^{\JG}$ must be winning for \eve
	from $(p',p'')$ with initial credit $c_1 = c + w(p,a,p') - w(t,a,p'')$.
	If $c_1 \leq |V|\maxm$, we are done.  Otherwise, by
	Proposition~\ref{pro:eg-props}, there is a strategy $\sigma'$ winning
	for \eve from $(p',p'')$ with initial credit $|V|\maxm$. From the
	definition of $W^{\JG}$ we get that $(p',p'') \in W^{\JG}$ as required,
	and the result follows.
\end{proof}

\begin{corollary}\label{cor:rel-two-strats}
	If there are winning strategies $\sigma^{\JG}_1,\sigma^{\JG}_2$ for \eve
	in the JG with initial credit $|V|\maxm$ from $(p,q_1),(p,q_2)$,
	respectively, such that
	$\sigma^{\JG}_1\left( (p,q_1)(p,q_1,a)\right) = (p,q_1,p_1,a)$ and
	$\sigma^{\JG}_2\left( (p,q_2)(p,q_2,a)\right) = (p,q_2,p_2,a)$
	for some $a \in A$,
	then $(p_1,p_2),(p_2,p_1) \in W^{\JG}$.
\end{corollary}

Finally, we note that by following a winning strategy for \eve in the JG, we are
sure all alternative runs of an automaton are related (by $W^{\JG}$)
to the run built by \eve.
\begin{lemma}\label{lem:good_strat}
	For all play prefixes $(p_0,q_0) (p_0,q_0,a_0) \dots$
	$(p_{n-1},q_{n-1},p_n,a_{n-1}) (p_n,q_n)$ consistent with a winning
	strategy for \eve in the JG from $(p_0,q_0)$ with initial credit
	$|V|\maxm$, for all runs $p_0 a_0 p'_1 a_1 \dots a_{n-1} p'_n$ of
	$\calN$ on $a_0 \dots a_{n-1}$ we have that $(p_n,p'_n) \in W^{\JG}$.
\end{lemma}
\begin{proof}
	Let $\sigma^{\JG}_1$ denote the winning strategy referred to in the claim.

	First, it is easy to show by induction that $(p_i,q_i) \in W^{\JG}$ for all
	$0 \le i \le n$. Intuitively, using the fact that
	$\sigma^{\JG}_1$ is winning for \eve
	with initial credit $|V|\maxm$ from $(p_0,q_0)$ we get that
	for any $(p_i,q_i)$ the strategy $\sigma^{\JG}_1$ is winning for her with
	some initial credit. Then, by Proposition~\ref{pro:eg-props}, there is
	another strategy $\sigma'$ that is winning from $(p_i,q_i)$ with initial
	credit $|V|\maxm$.

	We will now argue, by induction, that $(p_i,p'_i) \in W^{\JG}$ for all $0 <
	i \le n$. For the base case,
	it should be clear that $(p_1,p'_1) \in W^{\JG}$. This follows from
	Lemma~\ref{lem:joking}. Hence, we can assume the claim holds for some
	$0 < i < n$. By definition of $W^{\JG}$ we have that \eve has a winning
	strategy $\sigma^{\JG}_2$ in the JG from $(p_i,p'_i)$ with initial credit
	$|V|\maxm$. It follows from Corollary~\ref{cor:rel-two-strats}
	that $(p_{i+1},t) \in W^{\JG}$ where
	\(
		(p_i,q_i,p_{i+1},a_i) = 
		\sigma^{\JG}_1\left( (p_0,q_0) \dots (p_i,q_i,a_i)\right)
	\)
	and
	\(
		(p_i,q_i,t,a_i) =
		\sigma^{\JG}_2\left( (p_0,p_0) \dots (p_i,p'_i,a_i)\right).
	\)
	Using Lemma~\ref{lem:joking} we get that $(t,p'_{i+1}) \in W^{\JG}$. Now, by
	transitivity of $W^{\JG}$ (see Lemma~\ref{lem:transitivity}) we conclude
	that $(p_{i+1},p'_{i+1}) \in W^{\JG}$. The claim then follows by induction.
\end{proof}

We now prove that if \eve wins the JG played on $\calN$,
then the automaton $\calN$ is determinizable. In order to do so, we first prove that 
$\calN$ is $2|V|\maxm$-bounded.

\begin{proposition}\label{pro:JG_delay}
	Let $\calN = (Q,I,A,\Delta,w,F)$ be such that \eve wins the JG played on
	$\calN$. Then $\calN$ is $2|V|\maxm$-bounded.
\end{proposition}
\begin{proof}
	Let $\rho_p = p_0 a_0 p_1 \dots a_{n-1} p_n$ be a maximal accepting run
	of $\calN$, let $i \in \{0, \ldots, n\}$, and let $\rho_q = q_0 a_0 q_1
	\dots a_{i-1} q_i$ be an initial run.  Let us prove that
	\(
		w(\rho_q) - w(p_0 a_0 p_1 \dots a_{i-1} p_i) \leq 2|V|\maxm.
	\)
	Let $\rho_{p,1}$ denote the run $p_0 a_0 p_{1} \dots a_{i-1} p_{i}$, let
	$\rho_{p,2}$ denote the run $p_i a_i p_{i+1} \dots a_{n-1} p_{n}$.
	First, let $\sigma^{\JG}_1$ be a winning strategy for \eve in the JG from
	$(p'_0,q_0)$ (for some $p'_0 \in I$) with initial credit $|V|\maxm$. Let
	$\phi_q = (p'_0,q_0)(p'_0,q_0,a_0)(p'_0,q_0,p'_1,a_0)(p'_1,q_1) \dots
	(p'_i,q_i)$ be the play prefix consistent with $\sigma^{\JG}_1$  that
	corresponds to \adam playing the word $a_0 \ldots a_{i-1}$ and choosing
	the states from the run $\rho_q$.  Since $\sigma^{\JG}_1$ is winning,
	\adam cannot enforce a negative energy level, in other words
	$\el_{|V|\maxm}(\phi_q) \geq 0$, hence:
	\(
		w(\rho_q) - w(p'_0 a_0 p'_1 \dots a_{i-1} p'_i) \leq |V|\maxm.
	\)
	Second, by Lemma \ref{lem:good_strat}, $(p_i',p_i) \in W^{\JG}$, hence
	\eve has a winning strategy $\sigma^{\JG}_2$ in the JG starting from
	$(p_i',p_i)$ with initial credit $|V|\maxm$. Let $\phi_{p,2} =
	(p'_i,p_i) (p'_i,p_i,a_i) (p'_i,p_i,p'_{i+1},a_i) (p'_{i+1},p_{i+1})
	\dots (p'_n,p_n)$ be the play prefix consistent with $\sigma^{\JG}_2$
	that corresponds to \adam playing the word $a_i \ldots a_n$ and choosing
	the states from the run $\rho_{p,2}$.  Since $\sigma^{\JG}_2$ is
	winning, $p_n'$ is a final state and (for the same reason as above) we
	have
	\(
		w(\rho_{p,2}) - w(p'_i a_i p'_{i+1} \dots a_{n-1} p'_n)  \leq
		|V|\maxm = |V|\maxm.
	\)
	Finally, since  $\rho_p$ is maximal by hypothesis, $w(p'_0 a_0 p'_1
	\dots a_{n-1} p'_n) \leq w(\rho_p)$.  Since $w(\rho_p) = w(\rho_{p,1}) +
	w(\rho_{p,2})$, the desired result follows. 
\end{proof}

Since, by definition of the JG, both $|V|$ and $\maxm$ are polynomial w.r.t.
$|Q|$ and $\maxweight$, using Proposition~\ref{pro:bound-to-det} gives us the
following.
\begin{theorem}\label{thm:win-joker-det}
	Given an automaton $\calN = (Q,I,A,\Delta,w,F)$,
	if \eve wins the JG played on $\calN$, then there exists a
	deterministic automaton $\calD$ such that $\inter{\calD} =
	\inter{\calN}$, and whose size and maximal weight are
	polynomial w.r.t. $\maxweight$, and exponential w.r.t. $|Q|$.
\end{theorem}

With the results above, we are now in position to prove an
\EXP~upper bound for the $r$-regret problem.
\begin{proof}[Proof of Theorem~\ref{thm:regret-upper}]
	Given an automaton $\calN$ and $r \in \mathbb{N}$, we first determine
	whether \eve wins the JG played on $\calN = (Q,I,A,\Delta,w,F)$.
	To do so, we determine the winner of the corresponding EGR from all
	$(p,q) \in I^2$ with initial credit $|V|\maxm$. We can then, in
	polynomial time, decide if there exists $p \in I$ such that, for all $q
	\in I$, \eve can win the EGR from $(p,q)$. If the latter does not hold,
	then by contrapositive of Lemma~\ref{lem:joker-nec-regret}, $\calN$ is
	not $r$-regret determinizable. Otherwise, we construct $\calD$ such that
	$\inter{\calD} = \inter{\calN}$ and $\calD$ is deterministic using
	Theorem~\ref{thm:win-joker-det}. Finally, we use $\calD$ to decide if
	\eve wins the $r$-regret game using Proposition~\ref{pro:r-regret-eg}.
	Since $\calD$ is of size exponential w.r.t. to $\calN$ but its maximal
	weight is polynomial w.r.t. $\maxweight$, the resulting energy game
	without resets can be solved in exponential time by
	Proposition~\ref{pro:eg-props}.
\end{proof}

As a corollary, we obtain that the existential version of the $r$-regret problem
is also decidable. More precisely, using the techniques we have just presented,
we are able to decide the question: does there exist $r \in \mathbb{N}$ such
that a given automaton $\calN$ is $r$-regret determinizable? The algorithm to
decide the latter question is almost identical to the one we give for the
$r$-regret problem. The only difference lies in the last step, that is, the
energy game without resets constructed from the deterministic version of the
automaton that one can obtain from the JG. Instead of using a function of $r$ as
initial credit, we ask if \eve wins the energy game with initial credit
$|V|\maxm$---we also remove the gadget using vertices
$\bot^r_1,\bot^r_2$ which ensure a regret of at most $r$.
\begin{theorem}\label{thm:existential-regret-prob}
	Given an automaton, 
	deciding whether there exists $r \in \mathbb{N}$ such that it
	is $r$-regret determinizable is in \EXP.
\end{theorem}

\subsection{Lower bound}
In this section we argue that the complexity of the algorithm we described in
the previous section is optimal. More precisely, the $r$-regret problem
is \EXP-hard even if the regret threshold $r$ is fixed.

\begin{theorem}\label{thm:regret-lower-bound}
	Deciding the $r$-regret problem for a given automaton is
	\EXP-hard, even for fixed $r \in \mathbb{N}_{>0}$.
\end{theorem}

Observe that, in regret games, \eve may need to keep track of all runs of the
given automaton on the word $\alpha$ which is being ``spelled'' by \adam.
Indeed, if she has so far constructed the run $\rho$ and \adam chooses symbol
$a$ next, then her choice of transition to extend $\rho$ may depend on the set
of states at which alternative runs of the automaton on $\alpha$ end. The set of
all such configurations is exponential.
%

Our proof of the $r$-regret problem being \EXP-hard makes sure that \eve has to
keep track of a set of states as mentioned above. Then, we encode configurations
of a binary counter into the sets of states so that the set of states at which
\eve \emph{believes} alternative runs could be at, represent a valuation of the
binary counter. Finally, we give gadgets which simulate addition of constants to
the current valuation of the counter. These ingredients allow us to simulate
Countdown games~\cite{jsl08} using regret games. As the former kind of games
are~\EXP-hard, the result follows. The same reduction has been used to show that
regret minimization against \emph{eloquent adversaries} in several
\emph{quantitative synthesis games} is \EXP-hard~\cite{hpr15-journal}.

\section{Further research directions}
When the regret $r$ is given, the $r$-regret determinization problem
is \EXP-complete. When $r$ is not given, the problem is in
\EXP~but we did not found any lower bound other than
\P-hardness. Characterizing the precise complexity of this problem is open.

The latter is related to the following question. From our decision procedure for
solving the existential regret problem, it appears that if a WA is $r$-regret
determinizable for some $r$, it is also $r'$-regret determinizable for some $r'$
that depends exponentially on the WA. So far, we have not found any family of WA
that exhibit exponential regret behaviour, and the best lower bound we have is
quadratic in the size of the WA (see Appendix~\ref{sec:lower-bound-r}).

Finally, we would like to investigate the notions of delay- and
regret-determinization for other measures, such as discounted sum~\cite{cdh10}
or ratio~\cite{conf/concur/FiliotGR12}. These notions also make sense for other
problems, such as comparison and equality of weighted automata (which are
undecidable for max-plus automata), and disambiguation (deciding whether a given
WA is equivalent to some unambiguous one)~\cite{kirsten_et_al:LIPIcs:2009:1850}.

\bibliographystyle{alpha}
\bibliography{refs}

\clearpage
\appendix
\section{On Proposition~\ref{pro:0regret-easy}}
\myparagraph{Ratio vs. difference}
We remark that the results from~\cite{akl10} are on minimizing regret with
respect to the ratio measure and not the difference. However, ratio $1$
coincides with difference $0$.

\myparagraph{Determinization by pruning \& (N,min,+) vs. (Z,max,+)}
In~\cite{akl10} the authors actually consider automata over the \emph{tropical
semiring}. That is, their automata can only have non-negative integers
and they use $\min$ instead of $\max$ to aggregate multiple runs of the
automaton over the same word. However, it is easy to see that \eve wins an
$r$-regret game played on an automaton $\calN$ with integer weights if and
only if she wins an $r$-regret game on the same automaton with weights
``shifted'' so they all become negative (recall we use $\max$ and not $\min$).
More precisely we subtract $\maxweight$ from the weights of all transitions,
then multiply them all by $-1$, and denote the resulting automaton by $\calM$.
Clearly, the $(\mathbb{N}\cup \{ +\infty \},\min,+)$-automaton $\calM$ is
$r$-regret determinizable if and only if the $(\mathbb{Z}\cup \{ -\infty
\},\max,+)$-automaton $\calN$ is $r$-regret determinizable.

\section{Proof of Proposition~\ref{pro:bound-to-det}}

We present here in details the sketched construction.
Let us define $\mathcal{D} = (Q', \{q_I'\}, A,\Delta', F',w')$ as follows.
\begin{itemize}[nolistsep]
	\item $Q'$ is the set of functions from $Q$ to the set $\{-B, \ldots, B
		\} \cup \{ - \infty \}$.  The idea is that, on input $u$,
		$\mathcal{D}$ deterministically chooses a run $\rho$ of $\calN$ on $u$, outputs
		the corresponding weight, and uses its state to keep in memory,
		for each state $q \in Q$, the delay between $\rho$ and the
		maximal run of $\calN$ on $u$ ending in $q$;
	\item $q_I'$ is the function mapping each initial state of $\calN$ to $0$,
		and all the other states to $- \infty$;
	\item We now define $\Delta'$ and $w'$.  For every pair $(g,a) \in Q'
		\times A$, we have $(g,a,\delta_{g,a}) \in \Delta'$ and
		$w'(g,a,g') = w_{g,a}$, where $\delta_{g,a}$ and $w_{g,a}$ are
		constructed as follows.  First, we update the information
		concerning the runs of $\calN$ contained in $g$.  For every $q
		\in Q$, let 
		\[
			m_q = max\{ w + \mu \st (p,w) \in g, w(p,a,q) = \mu \}.
		\]
		The runs whose weight is too low are dropped.  Let $Q_B
		\subseteq Q$ be the set of states $q$ such that $m_p-m_q \leq B$
		for every $p \in Q$. In particular, if $m_q = - \infty$, $q \notin Q_B$.
		For every $q \in Q \setminus Q_B$, we set
		$\delta_{g,a}(q) = - \infty$.  Then, $w_{g,a}$ is defined as the
		weight corresponding to the maximal accepting run, if any is
		left, and to the value of the maximal (non accepting) run
		otherwise.  Formally, if $Q_B \cap F = \emptyset$, then $w_{g,a}
		= max\{ m_q \st q \in Q \}$, otherwise $w_{g,a} = max\{ m_q \st q
		\in F \}$.  Finally, the state is updated accordingly.  For
		every $q \in Q_B$, we set $\delta_{g,a}(q) = m_q-w_{g,a}$.
	\item $F'$ is the set of functions $g \in Q'$ such that $g(q_f) \neq -
		\infty$ for some final state $q_f \in F$.
\end{itemize}

By definition, $\mathcal{D}$ is complete and deterministic, $|Q'| = (2B+2)^{|Q|}$ and its maximal weight is 
$B + \maxweight$.
 In order to complete
the proof of Theorem  \ref{thm:win-joker-det}, we need to prove that
$\inter{\mathcal{D}} = \inter{\calN}$.  To do so, we expose three properties
satisfied by $\mathcal{D}$.

Given a run $\rho = p_0 a_0 p_1 \ldots a_{n-1} p_n$ of $\calN$, let us call
$\rho$ good if for every $0 \leq i \leq n$, and for every run $\rho_q = q_0 a_0
q_1 \dots a_{i-1} q_i$,
\[w(p_0 a_0 p_1 \dots a_{i-1} p_i) \geq w(\rho_q) - B.\]
Proposition \ref{pro:JG_delay} ensures us that every maximal accepting run of
$\calN$ is good.  Let $u \in A^*$, let $g_0 a_0 g_1 \ldots a_{n-1} g_n$ be the
run of $\mathcal{D}$ on $u$, and for every $0 \leq i \leq n$ let $w_i$ denote
the weight $w'(g_0 a_0 g_1 \ldots a_{i-1} g_i)$.

\begin{description}
	\item[\prop{1}:] For every $q \in F$, $g_n(q) \leq 0$, and if $g_n(p)
		\neq - \infty$ for at least one state $p \in F$, then there
		exists $q \in F$ such that $g_n(q) = 0$.
	\item[\prop{2}:] Let $q \in Q$.  If $g_n(q) \neq - \infty$, then there
		exists an initial run $\rho = q_0 a_0 q_1 \ldots a_{n-1} q_n$ of
		$\calN$ on $u$ such that $q_n = q$ and $w(\rho) = w_n + g_n(q)$.
	\item[\prop{3}:] Let $\rho = p_0 a_0 p_1 \ldots a_{n-1}p_n$ be an
		initial run of $\calN$ on $u$.  If $\rho$ is good, then $w_n +
		g_n(q_n) \geq w(\rho)$.
\end{description}

\begin{proof}
	\prop{1} follows immediately from the definition of $\Delta'$.  \prop{2}
	and \prop{3} are proved by induction on the size of $u$.  If $u =
	\epsilon$, then the state $g_n$ reached by $\mathcal{D}$ on input $u$ is
	the initial state $q_I'$ of $\mathcal{D}$, and the weight $w_n$
	corresponding to this run is $0$.  Then, by definition of $I'$, for
	every $q \in Q$ either $q$ is initial and $q_I'(q) = 0$, which is the
	value of the initial run $\rho' = q$ of $\calN$ on $\epsilon$, or $q$
	is not initial and $q_I'(q) = - \infty$. This proves \prop{2}.
	Conversely, every initial run $\rho$ of $\calN$ on $u$ is of the form
	$q_I$ for some $q_I \in I$, and $q_I'(q_I) = 0$, which proves \prop{3}.
	Now suppose that $u = va$ for some $v \in A^*$ and $a \in A$, and that
	\prop{2} and \prop{3} hold for $v$.

	Let us first prove that \prop{2} holds for $u$.  Suppose that $g_n(q)
	\neq - \infty$.  By definition of $\Delta'$, there exists $p \in Q$ and
	$(p,a,q) \in \Delta$ such that $g_n(q) = g_{n-1}(p) + w(p,a,q) -
	w'(g_{n-1},a,g_n)$.  Then $g_{n-1}(p) \neq - \infty$, hence, by the
	induction hypothesis, there exists an initial run $\rho' = q_0 a_0 q_1
	\ldots a_{n-2}q_{n-1}$ of $\calN$ on $v$ such that $q_{n-1} = p$ and
	$w(\rho') = w_{n-1} + g_{n-1}(p)$.  Then the run $\rho' a q$ of $\calN$
	on $u$ satisfies the statement of \prop{2}, since
	\[
	\begin{array}{lll}
	w(\rho' a q) & = & w(\rho') + w(p,a,q)\\
	& = & w_{n-1} + g_{n-1}(p) + w(p,a,q)\\
	& = & w_n - w'(g_{n-1},a,g_n) + { }\\
	&   & g_{n-1}(p) + w(p,a,q)\\
	& = & w_n + g_n(q).
	\end{array}
	\]

	Finally, we prove \prop{3}.  Suppose that $\rho$ is good. First, note
	that the run $p_0 a_0 p_1 \ldots a_{n-2}p_{n-1}$ on $v$, which is
	obtained by removing the last transition of $\rho$, is also good.
	Hence, by the induction hypothesis, \[ w_{n-1} + g_{n-1}(p_{n-1}) \geq
	w(\rho').  \] Second, if we suppose that $g_n(p_n) = - \infty$, we
	obtain a contradiction with the fact that $\rho$ is good, since, using
	\prop{2} and the definition of $\Delta'$, we are able to build a run
	$\psi$ on $u$ such that $w(\rho) < B - w(\psi)$.  Hence $g_n(p_n) \neq -
	\infty$, and, by definition of $\Delta'$, 
	\[
	w'(g_n a g_{n+1}) + g_n(p_n) \geq w(p_{n-1} a p_n) + g_{n-1}(p_{n-1}).
	\]
	These inequalities imply the correctness of \prop{3}, since
	\[
	\begin{array}{lll}
	w_n + g_n(p_n) & = & w_{n-1} +  w'(g_n a g_{n+1}) + g_n(p_n)\\
	& \geq & w_{n-1} + g_{n-1}(p_{n-1}) + w(q_{n-1} a p_n)\\
	& \geq & w(\rho') + w(q_{n-1} a p_n)\\
	& = & w(\rho).
	\end{array}
	\]
\end{proof}

\begin{corollary}
	The function defined by $\mathcal{D}$ is equal to $\inter{\calN}$.
\end{corollary}

\begin{proof}
	Let us begin by proving that for every input word $u$, $\inter{\calN}(u)
	\geq \inter{\mathcal{D}}(u)$.  Let $u \in \lang{\mathcal{D}}$, and let
	$g$ be the state reached by $\mathcal{D}$ on input $u$. Then $g$ is
	final, hence, by definition of $F'$, there exists a state $q_f \in F$
	such that $g(q_f) \neq - \infty$. Moreover, by \prop{1}, there exists a
	final state $p_f \in F$ such that $g_n(p_f) = 0$.  This implies, by
	\prop{2}, the existence of an initial run $\rho = p_0 a_0 p_1 \ldots
	a_{n-1} p_n$ of $\calN$ on $u$ such that $p_n = p_f \in F$ and $w(\rho)
	= \inter{\mathcal{D}}(u)$.  Since $q_n$ is a final state, $\rho$ is
	accepting, hence $\inter{\calN}(u) \geq w(\rho) =
	\inter{\mathcal{D}}(u)$.

	Finally, we prove that, conversely, for every input word $v$,
	$\inter{\mathcal{D}}(u) \geq \inter{\calN}(u)$.  Let $v \in
	\lang{\calN}$, let $\psi$ be a maximal accepting run of $\calN$ on $v$,
	and let $q \in F$ be the corresponding final state.  Let $g_v$ be the
	state of $\mathcal{D}$ reached on input $v$, and let $w_v$ be the
	associated output.  By Proposition \ref{pro:JG_delay}, $\psi$ is good,
	hence by \prop{3}, $w_v + g_v(q) \geq w(\psi) = \inter{\calN}(v)$.
	Therefore $g_v(q) \neq -\infty$, and, since $q$ is a final state of
	$\calN$, $g_v$ is a final state of $\mathcal{D}$.  Moreover, by
	\prop{1}, $g_v(q) \leq 0$, hence $\inter{\mathcal{D}}(u) = w_v \geq
	\inter{\calN}(u)$.
\end{proof}

\section{Proof of Lemma~\ref{pro:zero-delay-regret-pdet}}
\begin{proof}
	If $\calN$ is $0$-regret determinizable, then $\calN$ is $0$-delay
	determinizable by Proposition~\ref{pro:zero-delay-regret}. Now suppose
	that $\calN$ is $0$-delay determinizable, and let $\calD$ be a $0$-delay
	determinizer of $\calN$. Using $\calD$, we define a winning strategy
	$\sigma_{\calD}$ for \eve in the $0$-regret game played on $\calN$.
	Given a sequence $q_0 a_0 \ldots q_{n-1} a_{n-1}$ in $(Q \cdot A)^*$,
	the state $\sigma_{\calD}(q_0 a_0 \ldots q_{n-1} a_{n-1})$ is defined as
	follows. Let $\alpha$ denote the input word $a_0 \ldots a_{n-1}$.  If
	$\calD$ has no initial run on $\alpha$, this word is not a prefix of any
	word of $\lang{\calN}$, hence whatever \eve does, \adam will not be able
	to win.  We set $\sigma_{\calD}(q_0 a_0 \ldots q_{n-1} a_{n-1}) =
	q_{n-1}$.  Otherwise, let $\rho_{\alpha} = p_0 a_0 \ldots p_{n-1}
	a_{n-1} p_{n}$ be the initial run of $\calD$ on $\alpha$. Since $\calD$
	is a $0$-delay determinizer of $\calN$, there exists an initial run
	$\rho'_{\alpha} = p'_0 a_0 \ldots p'_{n-1} a_{n-1} p'_{n}$ of $\calN$
	such that for every $1 \leq i \leq n$, $w(p_0' \dots p_i') = w'(p_0
	\dots p_i)$.  Moreover, since $\calN$ is pair-deterministic,
	such a run is unique.  We set $\sigma_{\calD}(q_0 a_0 \ldots q_{n-1}
	a_{n-1}) = p'_{n}$.  Note that, for every $1 \leq i \leq n-1$, the run
	$\rho_{a_0 \ldots a_i}$ is equal to the prefix $p_0 a_0 \ldots p_{i}
	a_{i} p_{i+1}$ of $\rho_{\alpha}$, since $\calD$ is deterministic.
	Therefore, the run $\rho'_{a_0 \ldots a_i}$ is equal to the prefix $p'_0
	a_0 \ldots p'_{i} a_{i} p'_{i+1}$ of $\rho'_{\alpha}$, since $\calN$ is
	pair-deterministic.  Then $\sigma(p'_0 a_0 \ldots p'_{i} a_{i}) =
	p'_{i+1}$ for every $1 \leq i \leq n-1$, hence $\sigma(a_0 \ldots
	a_{n-1}) = \rho'$ and	
	\[
		\Val(\sigma(\alpha)) = \Val(\rho') = \Val(\rho) =
		\inter{\calD}(\alpha) = \inter{\calN}(\alpha).
	\]
	This proves that $\sigma_{\calD}$ is a winning strategy for the
	$0$-regret game played on $\calN$.
\end{proof}

\section{Making an automaton pair-deterministic}
\myparagraph{Subset construction}
Let $\calN =  (Q, I, A, \Delta, w, F)$ be an automaton.  Let
$\pow(\calN) =  (Q', I', A, \Delta', w', F')$ be the
automaton defined as follows.
\begin{itemize}[nolistsep]
	\item $Q' = \pow(Q)$;
	\item $I'= \{ I \} $;
	\item $\Delta' = \{ (U, a,\Delta_{a}^x(U)) \st a \in
		A, x \in \textsf{Im}(w) \}$, where $\Delta_{a}^x(U) =
		\bigcup_{p \in U}\{ q \in Q \st (p,a,q) \in \Delta,
		w(p,a,q) = x \}$;
	\item $w' : \Delta' \rightarrow \mathbb{Z}$, $(U, a,
		\Delta_{a}^x(U)) \mapsto x$;
	\item $F' = \{ P \subset Q \st P \cap F \neq \emptyset \}$.
\end{itemize}

\begin{lemma}\label{lem:subset-construction}
	The subset construction satisfies the following properties.
	\begin{enumerate}[nolistsep]
		\item\label{lem-item:sub_sub} $\calN \subseteq_0 \pow(\calN)$;
		\item\label{lem-item:sub_sup} $\pow(\calN) \subseteq_0 \calN$;
		\item\label{lem-item:sub_equ} $\inter{\pow(\calN)} =
			\inter{\calN}$.
	\end{enumerate}
\end{lemma}

\begin{proof}
	\item \ref{lem-item:sub_sub}.
	Let $\rho = q_0 a_0 \dots a_{n-1} q_n$ be an accepting
	run of $\calN$.  Let $U_0 = I$, and for every $1 \leq i \leq n$,
	let $U_i =
	\Delta_{a_{i-1}}^{w(q_{i-1},a_{i-1},q_i)}(U_{i-1})$.  Note
	that for every $0 \leq i \leq n$, $q_i \in U_i$.  Then $\rho' = U_0
	a_0 \dots a_{n-1} U_n$ is an accepting run of $\pow(\calN)$,
	and for every $0 \leq i < n$, $w'(U_{i},a_i,U_{i+1}) =
	w(q_{i},a_i,q_{i+1})$.  Therefore, $\calN \subseteq_0 \pow(\calN)$.

	\item \ref{lem-item:sub_sup}.
	Let $\rho = U_0 a_0 \dots a_{n-1} U_n$ be an
	accepting run of $\pow(\calN)$. Let $q_n$ be any element of $U_n
	\cap F$. For every $0 < i < n$, suppose that $q_{i+1} \in U_{i+1}$ is defined, and let $q_i \in U_i$ be inductively defined
	as follows.  By definition of $\Delta'$, $U_{i+1} =
	\Delta_{a_{i}}^{w'(U_{i},a_{i},U_{i+1})}(U_{i})$, hence, as
	$q_{i+1} \in U_{i+1}$, there exists $q_{i} \in U_{i}$ such that
	$(q_{i},a_{i},q_{i+1}) \in \Delta$ and
	$w(q_{i},a_{i},q_{i+1}) = w'(U_{i},a_{i},U_{i+1})$.  Then $\rho'
	= q_0 a_0 \dots a_{n-1} q_n$ is an accepting run of $\calN$,
	and $w(q_{i},a_i,q_{i+1}) = w'(U_{i},a_i,U_{i+1})$ for all $0
	\le i < n$.  Therefore, $\pow(\calN) \subseteq_0 \calN$.

	\item \ref{lem-item:sub_equ}.
	This property follows immediately from the two others and Lemma
	\ref{lem:k-inc_prop} item \ref{lem-item:k-inc_prop_sub}.
\end{proof}

\section{A lower bound on the required $r$ for $r$-regret determinizability}
\label{sec:lower-bound-r}
In this section we give
an example of an automaton which requires a
quadratic regret threshold $r$ for it to be $r$-regret determinizable.

\begin{proposition}\label{pro:quad-regret-necessary}
	Given an automaton $\calN$, a regret $r$ as big as
	$\mathcal{O}(|V|)$ might be needed for it to be $r$-regret
	determinizable.
\end{proposition}
\begin{proof}
	Let $k \in \mathbb{N}_{> 0}$ and consider the corresponding $\calN_k$
	automaton constructed as shown in Fig.~\ref{fig:quad-regret}. Note
	that $\calN_k$ consists of two deterministic automata.
	As the latter are also disjoint, the only decision for
	\eve to make in this game is to start from $p_1$ or from $q_1$. Further,
	notice that if she does start in $p_1$, then any word with more than $k$
	consecutive $a$'s forces her into the state $\bot$ which is not
	accepting. An alternative run starting from $q_1$ reaches $q_k$, which
	is accepting, when reading the same word. Thus, \eve really has no
	choice but to start in $q_1$ to realize at least the domain of $\calN_k$.

	Observe that any word with more than $k$ consecutive $a$'s is not
	accepted by the left sub-automaton. Hence, the maximal regret of the
	strategy for \eve which starts from $q_1$ is witnessed by a word of the
	form
	\(
		a^{i_1} b^{j_1} \dots a^{i_n} b^{j_n}
	\)
	where $i_\ell \le k$ for all $1 \le \ell \le n$. Such a word is assigned
	a value of $\sum_{\ell = 1}^n i_\ell$ by the left sub-automaton. On the
	other hand, any word with $k$ or more $b$'s is assigned a value
	equivalent to the length of the word minus $k$ by the right automaton.
	It is now easy to see that, if \eve starts in $q_1$, she will have
	a regret value of at least $k^2$. This value is realized, for instance,
	by the word $(a^kb)^k$.
\end{proof}

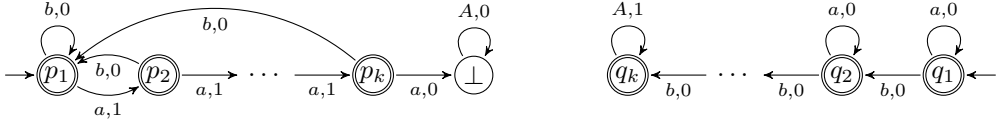
\begin{figure*}
\begin{center}
\begin{tikzpicture}[node distance=0.8cm]
	\node[state,initial left,accepting] (a1) {$p_1$};
	\node[state,right= of a1,accepting] (a2) {$p_2$};
	\node[right= of a2] (adots) {\dots};
	\node[state,right= of adots,accepting] (a3) {$p_k$};
	\node[state,right= of a3] (sink) {$\bot$};

	\node[state,right=1.5cm of sink,accepting] (b3) {$q_k$};
	\node[right= of b3] (bdots) {\dots};
	\node[state,right= of bdots,accepting] (b2) {$q_2$};
	\node[state,initial right,right= of b2,accepting] (b1) {$q_1$};

	\path
	(a1) edge[loopabove] node[el] {$b$,$0$} (a1)
	(a1) edge[bend right] node[el,swap] {$a$,$1$} (a2)
	(a2) edge[bend right] node[el] {$b$,$0$} (a1)
	(a2) edge node[el,swap] {$a$,$1$} (adots)
	(adots) edge node[el,swap] {$a$,$1$} (a3)
	(a3) edge node[el,swap] {$a$,$0$} (sink)
	(a3) edge[out=140,in=40] node[el] {$b$,$0$} (a1)
	(sink) edge[loopabove] node[el] {$A$,$0$} (sink)

	(b1) edge[loopabove] node[el] {$a$,$0$} (b1)
	(b1) edge node[el] {$b$,$0$} (b2)
	(b2) edge[loopabove] node[el] {$a$,$0$} (b2)
	(b2) edge node[el] {$b$,$0$} (bdots)
	(bdots) edge node[el] {$b$,$0$} (b3)
	(b3) edge[loopabove] node[el] {$A$,$1$} (b3)
	;
\end{tikzpicture}
\caption{Automaton $\calN_k$ with $2(k+1) + 1$ states and $A = \{a,b\}$.}
\label{fig:quad-regret}
\end{center}
\end{figure*}

\section{Proof of Proposition~\ref{pro:eg-props} (On Energy Games with Resets)}
In this section we study the properties of EGRs.  Let us
start by establishing that they are \emph{uniformly positionally determined}.
(The latter is even stronger than positional determinacy. A game is uniformly
positionally determined if, for all instances of the game, it is always the case
that both players have a positional strategy which is winning from all vertices
from which it is possible for that player to win, \ie~his winning region.) In
order to prove this, we will make use of a \emph{First Cycle Game}~\cite{ar14}.

To simplify our arguments, in the sequel we will fix an arbitrary unique initial
vertex $v_I$ from which the game starts.

\myparagraph{First cycle energy games} A \emph{first cycle energy game} (FCEG)
is played by \eve and \adam. Formally, an FCEG is---just like an EGR---a tuple
$\calG = (V,V_\exists,v_I,E,E_\rho,w)$ where the digraph $(V,E)$ has no sinks.
We call $\calG$ the \emph{arena} on which both an EGR or an FCEG could be
played.  The main difference between the FCEG and the EGR played on $\calG$ is
that the former is a \textbf{finite} game. More precisely, the FCEG is played up
to the point when the first cycle is formed.  The winner of the game is then
determined by looking at the cycle: if it has a negative sum of weights and it
does not contain a reset edge, then \adam wins; else \eve wins. In other words,
\adam and \eve choose edges, from $v_I$ to form a \emph{lasso}, \ie~a finite
path $\phi = v_0 \dots v_i \dots v_n$ such that $v_n = v_i$ and $v_j \neq v_k$
for all $0 \le j < k < n$. We then say that $\phi$ is winning for \eve if and
only if $(v_j,v_{j+1}) \in E_\rho$ for some $i \le j < n$ or
$\sum_{\ell=i}^{n-1} w(v_\ell,v_{\ell+1}) \ge 0$, otherwise $\phi$ is winning
for \adam. (Note that the property of $\phi$ being winning for \eve is
determined solely on the ``cycle part'' of the lasso. More specifically, the
cycle $v_i \dots v_n$.)

\myparagraph{EGRs are greedy}
Let $Y$ be a \emph{cycle property}. We say
an infinite-duration game is $Y$-greedy (or just greedy, when $Y$ is clear from
the context) if and only if:
\begin{itemize}
	\item all plays $\pi$ such that every cycle in $\pi$ satisfies $Y$ are
		winning for \eve; and
	\item all plays $\pi$ such that every cycle in $\pi$ does not satisfy
		$Y$ are winning for \adam.
\end{itemize}

We will now focus on the cycle property used to determine if \eve wins the FCEG
defined above: either the cycle contains a reset edge or the sum of the weights of
cycle is non-negative. Let us start with the following observation.
\begin{lemma}\label{lem:adam-wins-any}
	If a play $\pi$ is such that all of its cycles have negative sum of
	weights, then for all $c_0 \in \mathbb{N}$ there is a finite path $\phi$ which
	is a prefix of $\pi$ and for which it holds that $\el_{c_0}(\phi) < 0$.
\end{lemma}
We claim that EGRs are greedy with respect to
this property.
\begin{lemma}\label{lem:eg-greedy}
	An EGR $\calG = (V,V_\exists,v_I,E,E_\rho,w)$ with initial credit
	$|V|\maxweight$ is greedy.
\end{lemma}
\begin{proof}
	For convenience, we will focus on simple cycles.

	Let $\pi = v_0 v_1 \dots$ be a play of the game such that all cycles
	from $\pi$ either contain a reset edge or have non-negative sum of weights.
	We will argue that, for all prefixes $\phi$ of $\pi$ with length at
	most $n$ we have that $\el_{c_0}(\phi) \ge 0$. We proceed by induction on $n$.
	If $n = 0$ then the claim holds trivially. Now, let us consider an
	arbitrary prefix $\phi = v_0 \dots v_{n+1}$. By induction
	hypothesis, we have that $\el_{c_0}(v_0 \dots v_n) \ge 0$.  Let
	$i$ be the index of the latest occurrence of a reset edge in $\phi$ (with $i
	= 0$ if there is no such edge). If $\el_{c_0}(\phi) < 0$ then clearly
	$n - i > |V|$ since at least $|V|$ edges are necessary to go from
	$c_0 = |V|\maxweight$ to a negative number in $\calG$. Furthermore, it
	follows that between $i$ and $n$ we have cycle $\chi = v_i \dots
	v_j$ which does not contain a reset edge and such that the sum of
	its transition weights is negative. This contradicts our assumption that
	all cycles from $\pi$ have a reset edge or non-negative sum of weights.
	Hence, the claim holds by induction and the play is winning for \eve
	with initial credit $|V|\maxweight$.

	Let $\pi = v_0 v_1 \dots$ be a play of the game such that all cycles
	from $\pi$ have negative sum of weights. From
	Lemma~\ref{lem:adam-wins-any} we have that for some prefix $\phi$ of
	$\pi$ the energy level becomes negative, \ie~$\el_{c_0}(\phi) < 0$. Hence, the
	play is winning for \adam.	
\end{proof}

Since EGRs are greedy, it follows from~\cite{ar14} that strategies transfer
between an EGR and an FCEG played on the same arena. More formally,
\begin{proposition}\label{pro:strat-transfer}
	Let $\calG = (V,V_\exists,v_I,E,E_\rho,w)$ be an arena.
	Every memoryless strategy $s$ for \eve (\adam) in the ERG played
	on $\calG$ with initial credit $|V|\maxweight$ is winning for her (him)
	if and only if $s$ is winning for her (him) in the FCEG played on
	$\calG$.
\end{proposition}

\myparagraph{More cycle properties} Let $Y$ be a cycle property. We say $Y$ is
\emph{closed under cyclic permutations} if for any cycle $\chi = v_0 v_1 \dots
v_{n-1}v_0$ such that $\chi \models Y$ we have that $v_1 v_2
\dots v_{n-1} v_0 v_1  \models Y$.
We also say $Y$ is \emph{closed under concatenation} if for any
two cycles $\chi = v_0 v_1 \dots v_{n-1} v_0$ and $\chi' = v_0 v'_1 \dots
v'_{m-1} v_0$ such that $\chi,\chi' \models Y$ we have that $\chi \dot v'_1 \dots
v'_{m-1} v_0 \models Y$.
\begin{lemma}\label{lem:eg-closed}
	The property of a cycle being winning for \eve (\adam) in an FCEG is
	closed under cyclic permutations and concatenation.
\end{lemma}
\begin{proof}
	Clearly, a cycle being winning for \adam is closed under both
	operations. Indeed, by commutativity of addition, if a cycle has
	negative sum of weights, the order of the weights does not matter.
	Additionally, the concatenation of two negatively weighted
	cycles concatenated yields a
	negatively weighted cycle.

	Let us consider now the property of a cycle being winning for \eve. If
	it contains a reset edge then any cyclic permutation of the cycle will
	also contain it. If it does not have a reset edge then it must have a
	non-negative sum of weights. Thus, the result follows by commutativity
	of addition. Additionally, concatenation preserves containment of
	a reset edge and, again, two positively weighted cycles can only yield a
	positively weighted cycle when concatenated.
\end{proof}
It then follows immediately from Lemmas~\ref{lem:eg-greedy},~\ref{lem:eg-closed}
and~\cite{ar14} that EGRs and FCEGs are uniformly positionally determined.
\begin{proposition}\label{pro:eg-pos-det}
	Let $\calG = (V,V_\exists,v_I,E,E_\rho,w)$ be an arena.
	The EGR played on $\calG$ with initial credit $|V|\maxweight$ and
	the FCEG played on $\calG$ are both uniformly positionally determined.
\end{proposition}

Using Lemma~\ref{lem:adam-wins-any} and Proposition~\ref{pro:eg-pos-det} we can
show that \eve wins an EGR with some initial credit $c_0 \in \mathbb{N}$ if and
only if she wins with initial credit $|V|\maxweight$.
\begin{proposition}\label{pro:enough-credit}
	Let $\calG = (V,V_\exists,v_I,E,E_\rho,w)$ be an arena
	\eve wins the EGR played on $\calG$ with initial credit
	$|V|\maxweight$ if and only if she wins it with some initial credit $c_0
	\in \mathbb{N}$.
\end{proposition}
\begin{proof}
	One direction is obvious: if \eve wins the EGR with initial credit $c_0 =
	|V|\maxweight$ then clearly she wins the EGR with some initial credit. We
	argue that if there exists some initial credit with which she wins the
	EGR then $|V|\maxweight$ suffices. We will, in fact, show that the
	contrapositive holds. Suppose that \eve does \textbf{not} win the EGR
	with initial credit $|V|\maxweight$. Then by determinacy of the EGR
	with that initial credit (Proposition~\ref{pro:eg-pos-det}), and using
	Proposition~\ref{pro:strat-transfer} together with
	Lemma~\ref{lem:adam-wins-any} we have that \adam has a strategy $\tau$
	in the EGR which, regardless of the initial credit, ensures a negative
	energy level is witnessed. Hence, there is no initial credit for which
	\eve wins.
\end{proof}

\myparagraph{A pseudo-polynomial algorithm} We will reduce the problem of
deciding if \eve wins an EGR with a given initial credit $c_0 \in \mathbb{N}$ to
that of deciding if she wins a \emph{safety game}~\cite{ag11}. Safety games are
played by \eve and \adam on an unweighted arena $(V,V_\exists,v_I,E,U)$ with a
set $U \subseteq V$ of \emph{unsafe vertices} which determine
the goals of the players. \eve wins the safety game if she has a strategy which
ensures no the play does not contain vertices from $U$, otherwise \adam
wins. Safety games are known to be uniformly positionally determined and
solvable in linear time~\cite{ag11}.

More formally, for a given EGR $(V,V_\exists,v_I,E,E_\rho,w)$
and initial credit $c_0$, we define a safety game played on
$(V',V'_\exists,v'_I,E',U)$ where
\begin{itemize}[nolistsep]
	\item $V' = V \times \left(\{\bot\} \cup \{0, 1, \dots,
		|V|\maxweight\}\right)$,
	\item $V'_\exists = V \times \left(\{\bot\} \cup
		\{0,1,\dots,|V|\maxweight\} \right)$,
	\item $v'_I = (v_I,c_0)$,
	\item $E'$ includes the edge $( (u,c), (q,v) )$ if and only
		if
		\begin{itemize}[nolistsep]
			\item $(u,v) \in E \setminus E_\rho$, $c \neq \bot$,
				and $d = \max\{w(u,v) + c,
				|V|\maxweight\}$ or
			\item $(u,v) \in E_\rho$, $c \neq \bot$, 
				and $d = \max\{w(u,v) + c_0, |V|\maxweight$ or
			\item $c = \bot$, $d = \bot$, and $u = v$.
		\end{itemize}
	\item $U = V \times \{\bot\}$.
\end{itemize}
Informally, if the energy level goes above $|V|\maxweight$ then we ``bounce'' it
back to $|V|\maxweight$.

\begin{proposition}\label{pro:eg-to-safe}
	\eve wins the EGR $(V,V_\exists,v_I,E,E_\rho,w)$
	with initial credit $c_0 \in \mathbb{N}$ if and only if
	she wins the safety game played on $(V',V'_\exists,v'_I,E',U)$.
\end{proposition}
\begin{proof}
	Clearly, if \eve wins the safety game then she wins the EGR with initial
	credit $c_0$.
	
	Conversely, if she has a strategy $\sigma$ to win
	the EGR with initial credit $c_0$, then, by
	Proposition~\ref{pro:eg-pos-det} she can do so with a uniformly
	positional strategy. That is, she has a strategy $\sigma'$ which ensures
	that from any vertex from which she can win with some initial credit,
	she wins with $\sigma'$ and initial credit $|V|\maxweight$.
	We claim that \eve must also be able to win the safety game.
	Indeed, at least until the first time the energy level is ``bounced''
	back, by playing according to $\sigma$, the energy level cannot become
	negative. For any play which does reach a point at
	which the energy level is ``bounced'' back, we observe that the reached
	vertex $u$ (the first component of the safety-game vertex $(u,c)$) is also
	reachable in the EGR by a play consistent with $\sigma$. Hence, $\sigma$
	must be winning for \eve in the EGR played from $u$ with some initial
	credit. By Proposition~\ref{pro:enough-credit}, $|V|\maxweight$ should
	suffice for $\sigma'$ to be winning for her in the EGR from $u$.
	Henceforth, she plays according to $\sigma'$ and the energy level cannot
	become negative by the above argument.

	The result then follows by determinacy of safety games.
\end{proof}

\section{Proof of Proposition~\ref{pro:r-regret-eg}}
Given an automaton $\calN = (Q,I,A,\Delta,w,F)$ and $\calD =
(Q',\{q'_I\},A,\Delta',w',F')$ such that $\calD$ is deterministic and
$\inter{\calD} = \inter{\calN}$, we construct the energy game without resets is
$\calG = (V,V_\exists,E, \emptyset,\mu)$ as described in the sketch provided in
the main body of the paper.

Since $\calG$ is of size polynomial w.r.t. $\calD$ and $\calN$, and we have a
pseudo-polynomial algorithm to determine the winner of energy games, it suffices
for us to prove the following claim.
\begin{lemma}
	The automaton $\calN$ is $r$-regret determinizable if and only if
	there exists $p_I \in I$ such that \eve wins the energy game without
	resets $\calG$ from $(p_I,q'_I)$ with initial credit $r +
	|Q'|(\maxweight + \maxweight')$.
\end{lemma}
\begin{proof}
	We will argue that if \eve wins the energy game, then she wins the
	$r$-regret game and if \adam wins the energy game, then she cannot win
	the $r$-regret game. The desired result follows from determinacy of
	energy games (Proposition~\ref{pro:eg-props}).

	Assume \eve wins the game from some $(p_I,q'_I)$ with strategy $\sigma$.
	Clearly, any play consistent with $\sigma$ never reaches the vertex
	$\bot$. The strategy $\sigma$ can be turned into a strategy $\sigma'$
	for \eve in the regret game as follows: for every symbol given by \adam
	in the regret game, $\sigma'$ selects a transition of $\calN$ based on
	what $\sigma$ does in response to the deterministic transition of
	$\calD$. More formally, for any word $\alpha = a_0 \dots a_{n-1} \in
	A^*$ which can be extended to a word $\alpha' \in \lang{\calN}$, we have
	$\sigma'(\epsilon) = p_I$ and $\sigma'(\alpha) = \sigma\left(
	(p_0,q_0) (p_0,q_0,q_1,a_0) \dots (p_{n-1},q_{n-1},q_n,a_{n-1}) \right)$
	where $p_0 = p_I, q_0 = q'_I$ and
	\[
		(p_0,q_0) (p_0,q_0,q_1,a_0) \dots (p_{n-1},q_{n-1},q_n,a_{n-1})
	\]
	is consistent with $\sigma$. The latter
	is well defined since we have argued that no play consistent with
	$\sigma$ reaches $\bot$. Additionally, since \adam can choose to avoid
	$\top^r_1,\top^r_2$, there are plays consistent with any strategy of
	\eve which do not reach these vertices. Finally, since we have assumed
	$\alpha$ can be extended to a word in the language of $\calN$, $\top$
	cannot be reached. For words which cannot be extended in this way,
	$\sigma'$ behaves arbitrarily. Observe that if $\alpha \in
	\lang{\calN}$, then $\alpha \in \lang{\calD}$ and thus $p_n \in F$ since
	otherwise \adam could reach $\bot$ in the energy game when playing
	against $\sigma$, and this would contradict the fact that $\sigma$ is
	winning. Furthermore, we have that $\inter{\calN}(\alpha) -
	\Val(\sigma'(\alpha)) \le r$ since otherwise \adam could reach
	$\bot^r_1$ in the energy game when playing against $\sigma$ and make her
	lose the game (since $\inter{\calN}(\alpha) = \inter{\calD}(\alpha) =
	w(q_0 \dots q_n)$), again contradicting the fact that $\sigma$ is
	winning not be winning.

	Assume \adam wins the game $\calG$ from every $(p_I,q'_I)$. Suppose, for
	a contradiction, that \eve has a strategy $\sigma$ with which she wins
	the $r$-regret game. Let $\sigma(\epsilon) = p_0$ and $\tau$ be the
	strategy for \adam in the energy game which is winning for him from
	$(p_0,q'_I)$. The strategy $\sigma$ can be turned into a strategy for
	\eve in the energy game by ignoring the states of $\calD$ and choosing
	transitions of $\calN$ when \adam chooses a symbol. Hence, the play
	$\pi^{(p_0,q'_0)}_{\sigma \tau}$ in the energy game must be losing for
	\eve, by choice of $\tau$. If the play is losing because vertex $\bot$
	is reached, then either $\sigma$ does not reach a final state of $\calN$
	after reading a word in $\lang{\calN}$ or she got stuck and cannot
	continue choosing transitions. In both cases, this contradicts the fact
	that $\sigma$ is winning for her in the regret game. If the play is
	losing because vertices $\bot^r_1$ and $\bot^r_2$ were reached, then
	$\sigma$ cannot be winning in the regret game because it assigns a
	weight to a word $\alpha \in \lang{\calN}$ that is too low w.r.t. to
	$\inter{\calD}(\alpha) = \inter{\calN}(\alpha)$. Finally, if the play is
	losing because the energy level drops below $0$, despite the initial
	credit of $r + |Q'|(\maxweight +
	\maxweight')$, after prefix $\phi = \dots (p_n,q_n)$, then since $\calD$
	is trim, there is a word $\beta = b_0 \dots b_{m-1}$ such that $m \le
	|Q'|$ and there is an accepting of $\calD$ from $q_n$ on $\beta$.
	Clearly then, the difference between the value of the run constructed
	by $\sigma$ and the value assigned by $\calD$ to the overall word is
	strictly greater than $r$. Hence, $\sigma$ cannot be winning for \eve
	in the regret game.
\end{proof}

\section{Proof of Lemma~\ref{lem:transitivity}}

\begin{proof} 
	Let $W'$ denote the set of pairs $(p,t) \in Q^2$ such that $(p,q),(q,t)
	\in W^{\JG}$ for some $q \in Q$.  Given $(p,t) \in W'$, let $q_{p,t} \in
	Q$ denote the state such that $\credit(p,q_{p,t}) + \credit(q_{p,t},t)$
	is minimal.  We prove that for every $(p,t) \in W'$, $\credit(p,t) \leq
	\credit(p,q_{p,t}) + \credit(q_{p,t},t)$.  The lemma then follows, by
	Proposition~\ref{pro:eg-props} and the definition of $W^{\JG}$.
	
	The proof is done by exposing a positional strategy $\sigma$ for \eve in
	the JG played on $\calN = (Q,I,A,\Delta,w,F)$.
	For every $(p,q) \in W^{\JG}$, let
	$\sigma_{(p,q)}$ denote a winning strategy for \eve in the JG played
	from $(p,q)$ with initial credit $\credit(p,q)$.  For every $(p,t,a) \in
	Q^2 \times A$, if $(p,t) \in W'$, let $\sigma((p,t,a)) =
	\sigma_{(p,q_{p,t})}((p,q_{p,t})(p,q_{p,t},a))$.  We now prove that, for
	every $p,t \in W'$, the strategy $\sigma$ is winning for \eve in the JG
	played on $\calN$ starting from $(p,t)$, with initial credit
	$\credit(p,q_{p,t}) + \credit(q_{p,t},t)$.  Suppose, towards a
	contradiction, that there exists a play consistent with $\sigma$
	\[
		\phi = (p_0, t_0) (p_0, t_0, a_0) (p_0, t_0,p_1, a_0)(p_1,t_1)
		\dots (p_n,t_n)
	\]
	such that $\el_{c_0}(\phi) < 0$, where $c_0 = \credit(p_0,q_{p_0,t_0}) +
	\credit(q_{p_0,t_0},t_0)$.  Moreover, let us suppose that we have chosen
	the play $\phi$ such that there exists no play of shorter length
	satisfying this property.
	
	First, let us set
	\[
		\phi' = (p_1, t_1) (p_1, t_1, a_1) (p_1, t_1,p_2, a_1)(p_2,t_2) \dots
		(p_n,t_n)
	\]
	and $c_1 = \credit(p_1,q_{p_1,t_1}) + \credit(q_{p_1,t_1},t_1)$.  By the
	hypothesis of minimality over the length of $\phi$, we obtain that
	$\el_{c_1}(\phi') \geq 0$.  Now, let $q_0 = q_{p_0,t_0}$ and let $q_1' =
	\sigma_{(q_0,t_0)}((p_0,t_0,a))$.  Note that, by definition of $\sigma$,
	$p_1 = \sigma_{(p_0,q_0)}((p_0,q_0,a))$.  Therefore, by Lemma
	\ref{lem:joking}, $\credit(p_1,q'_1) \leq \credit(p_0,q_0) +
	w(p_0,a,p_1) - w(q_0,a,q'_1)$ and $\credit(q'_1,t_1) \leq
	\credit(q_0,t_0) + w(q_0,a,q'_1) - w(t_0,a,t_1)$.  Moreover, by
	definition of $q_{p_1,t_1}$, $c_1 \leq \credit(p_1,q'_1) +
	\credit(q'_1,t_1)$, hence
	\[
		c_1 \leq c_0 + w(p_0,a,p_1) - w(t_0,a,t_1).
	\]
	Therefore, we obtain
	\[
		\begin{array}{lll}
		\el_{c_0}(\phi) & = & \el_{0}(\phi) + c_0\\
		& = & \el_{0}(\phi') + c_0 + w(p_0,a,p_1) - w(t_0,a,t_1)\\
		& \geq & c_0 - c_1 + w(p_0,a,p_1) - w(t_0,a,t_1)\\
		& \geq & 0,
		\end{array}
	\]
	which is a contradiction.
\end{proof}

\section{Proof of Theorem~\ref{thm:existential-regret-prob}}
We will adapt the proof of Proposition~\ref{pro:r-regret-eg} to show that the
existential $r$-regret problem can be solved by reduction to an energy game if a
deterministic version of the automaton is known. Together with the techniques
developed in Section~\ref{sec:upper-regret}, this will imply the existential
$r$-regret problem is decidable in exponential time.

\begin{proposition}\label{pro:exist-regret-eg}
	Given an automaton $\calN = (Q,I,A,\Delta,w,F)$ and $\calD =
	(Q',\{q'_I\},A,\Delta',w',F')$ such that $\calD$ is deterministic and
	$\inter{\calD} = \inter{\calN}$, the existential $r$-regret problem for
	$\calN$ is decidable in time polynomial in $|Q|$, $|Q'|$,  $|A|$, 
	$\maxweight$, and $\maxweight'$.
\end{proposition}
\begin{proof}
	As for the proof of Proposition~\ref{pro:r-regret-eg}, we
	construct an energy game without resets
	which simulates the regret game played on $\calN$
	while using $\calD$ to compare the weights of transitions chosen by \eve
	to those of the maximal run of $\calN$. Crucially, we will not add
	gadgets to punish \eve if she does not ensure a regret of at most $r$.
	This is because we do not fix such an $r$ a priori.
	Formally, the energy
	game without resets is $\calG = (V,V_\exists,E, \emptyset,\mu)$ where:
	\begin{itemize}[nolistsep]
		\item $V = Q^2 \cup Q^3 \times A \cup \{\top,\bot\}$;
		\item $V_\exists =  Q^3 \times A$;
		\item $E$ contains edges to simulate transitions of $\calN$
			and $\calD$, \ie 
			$\{ \left((p,q),(p,q,q',a)\right) \st (q,a,q') \in
			\Delta' \} \cup \{ \left((p,q,q',a),(p',q')\right) \st
			(p,a,p') \in \Delta \}$, edges required
			to verify \eve does not reach a non-final state when
			$\calD$ accepts, \ie~$\{ \left((p,q),\bot\right) \st p
			\not\in F \land q \in F' \} \cup \{ (\bot,\bot) \}$,
			and edges to punish one of the
			players if an automaton blocks, \ie~$\{ \left( (p,q), \top
			\right) \st \lnot \exists (q,a,q') \in \Delta'\} \cup
			\{ \left( (p,q,q',a), \bot \right) \st \lnot \exists
			(p,a,p') \in \Delta\} \cup \{(\top,\top)\}$;
		\item $\mu : E \to \mathbb{Z}$ is such that
			\begin{itemize}[nolistsep]
				\item $\left( (p,q,q',a),(p',q') \right) \mapsto
					w(p,a,p') - w'(q,a,q')$,
				\item $(\bot,\bot) \mapsto -1$,
				\item $(\top,\top) \mapsto 1$, and
				\item $e \mapsto 0$ for all other $e \in E$.
			\end{itemize}
	\end{itemize}
	We then claim that for some $p_I \in I$, \eve wins the energy game
	without resets $\calG$ from $(p_I,q'_I)$ with initial credit 
	$|V|\maxm$ if and only if $\calN$ is $r$-regret
	determinizable for some $r \in \mathbb{N}$.
	As before, the result will follow from the fact $\calG$ is of size
	polynomial w.r.t. $\calD$ and $\calN$, and the application of the
	pseudo-polynomial algorithm to determine the winner of $\calG$.

	Assume \eve wins the game from some $(p_I,q'_I)$ with strategy $\sigma$.
	Clearly, any play consistent with $\sigma$ never reaches the vertex
	$\bot$. The strategy $\sigma$ can be turned into a strategy $\sigma'$
	for \eve in the regret game as follows: for every symbol given by \adam
	in the regret game, $\sigma'$ selects a transition of $\calN$ based on
	what $\sigma$ does in response to the deterministic transition of
	$\calD$. More formally, for any word $\alpha = a_0 \dots a_{n-1} \in
	A^*$ which can be extended to a word $\alpha' \in \lang{\calN}$, we have
	$\sigma'(\epsilon) = p_I$ and $\sigma'(\alpha) = \sigma\left(
	(p_0,q_0) (p_0,q_0,q_1,a_0) \dots (p_{n-1},q_{n-1},q_n,a_{n-1}) \right)$
	where $p_0 = p_I, q_0 = q'_I$ and 
	\[
		(p_0,q_0) (p_0,q_0,q_1,a_0) \dots (p_{n-1},q_{n-1},q_n,a_{n-1})
	\]
	is consistent with $\sigma$. The latter
	is well defined since we have argued that no play consistent with
	$\sigma$ reaches $\bot$. Also, since we have assumed
	$\alpha$ can be extended to a word in the language of $\calN$, $\top$
	cannot be reached. For words which cannot be extended in this way,
	$\sigma'$ behaves arbitrarily. Observe that if $\alpha \in
	\lang{\calN}$, then $\alpha \in \lang{\calD}$ and thus $p_n \in F$ since
	otherwise \adam could reach $\bot$ in the energy game when playing
	against $\sigma$, and this would contradict the fact that $\sigma$ is
	winning. Furthermore, we have that $\inter{\calN}(\alpha) -
	\Val(\sigma'(\alpha)) \le |V|\maxm$, since $\sigma$
	is winning for eve in the energy game. Hence, \eve wins the $r$-regret
	game for $r = |V|\maxm$.

	Suppose \eve does not win the game from some $(p_I,q'_I)$ with initial
	credit $|V|\maxm$, then by
	Proposition~\ref{pro:eg-props}, \adam wins the game for every
	$(p_I,q'_I)$ regardless of the initial credit $c_0$.  Suppose, for a
	contradiction, that \eve has a strategy $\sigma$ with which she wins the
	$r$-regret game for some $r \in \mathbb{N}$. Because of our reduction
	from $r$-regret games to EGRs, we obtain from
	Proposition~\ref{pro:eg-props} that $\sigma$ can be assumed to be a
	finite memory strategy. Let $m_\sigma$ be the amount of memory used by
	the machine realizing $\sigma$, \ie~the size of the machine.
	Let $\sigma(\epsilon) =
	p_0$ and $\tau$ be the strategy for \adam in the energy game which is
	winning for him from $(p_0,q'_I)$. The strategy $\sigma$ can be turned
	into a strategy for \eve in the energy game by ignoring the states of
	$\calD$ and choosing transitions of $\calN$ when \adam chooses a symbol.
	Hence, the play $\pi^{(p_0,q'_0)}_{\sigma \tau}$ in the energy game must
	be losing for \eve, by choice of $\tau$. If the play is losing because
	vertex $\bot$ is reached, then either $\sigma$ does not reach a final
	state of $\calN$ after reading a word in $\lang{\calN}$ or she got stuck
	and cannot continue choosing transitions. In both cases, this
	contradicts the fact that $\sigma$ is winning for her in the regret
	game. Now, we will focus on the case where 
	the play is losing because the energy level drops
	below $0$. Recall that this will be the case, regardless of the initial
	credit, by choice of $\tau$. Thus, let $c_0 = |V|\maxm m_\sigma$ and
	let $\phi = \dots (p_n,q_n)$ be the minimal prefix of $\pi^{(p_0,q'_0)}_{\sigma
	\tau}$ such that $\el_{c_0}(\phi) < 0$. Clearly $\phi$ contains a
	negatively-weighted cycle $\chi$ which, furthermore, is a cycle on the
	machine realizing $\sigma$. Hence, \adam can ``pump'' $\chi$. A bit more
	precisely: after $\phi$, \adam can repeat $r + |Q'|(\maxweight +
	\maxweight')$ times the cycle $\chi$ and then spell any word which will
	make $\calN$ accept the word (recall $\calN$ is trim, so this is always
	possible from every state) and make sure the regret of $\sigma$ is
	greater than $r$. The latter contradicts the fact that $\sigma$ is
	winning for \eve in the $r$-regret game. The result thus follows.
\end{proof}

\section{Proof of Theorem~\ref{thm:regret-lower-bound}}
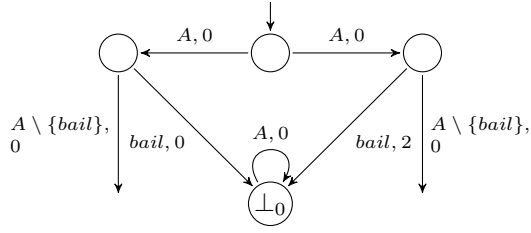
\begin{figure}
\begin{center}
\begin{tikzpicture}
\node[state] (sink0) at (2,0) {{$\bot_0$}};
\node[state] (linter) at (0,2) {};
\node[state,initial above] (vi) at (2,2) {};
\node[state] (rinter) at (4,2) {};
\node (lbinter) at (0,0) {};
\node (rbinter) at (4,0) {};

\path
(sink0) edge[loopabove] node[el] {$A,0$} (sink0)
(vi) edge node[el,swap] {$A,0$} (linter)
(linter) edge node[el,swap] {$bail, 0$} (sink0)
(vi) edge node[el] {$A,0$} (rinter)
(rinter) edge node[el] {$bail, 2$} (sink0)
(linter) edge node[el,swap,align=left] {$A \setminus \{bail\}$,\\$0$} (lbinter)
(rinter) edge node[el,align=left] {$A \setminus \{bail\}$,\\$0$} (rbinter)
;
\end{tikzpicture}
\caption{Initial gadget used in reduction from countdown games.}
\label{fig:initial-gadget}
\end{center}
\end{figure}

\begin{figure*}
\begin{center}
\begin{tikzpicture}
\node at (-0.2,2) {\dots};
\node[state] (xi) at (4,0) {{$x_i$}};
\node[state] (notxi) at (4,4) {{$\overline{x_i}$}};
\node[state] (carryi) at (3,2) {};
\node[state] (sink2) at (0.5,3.5) {{$\bot_2$}};
\node at (6.5,2) {\dots};
\node[state] (xn) at (8,0) {{$\overline{x_n}$}};
\node[state] (notxn) at (8,4) {{$x_n$}};
\node[state] (sink22) at (10,4) {{$\bot_2$}};

\path
(xi) edge[out=190,in=-100] node[el]{$c_{i+1},2$} (sink2)
(xi) edge[bend left] node[swap,el, align=left]{$b_i,0$\\$c_i,0$} (carryi)
(carryi) edge node[el,pos=0.05] {$A \setminus \{c_{i+1}\},2$} (sink2)
(carryi) edge[bend left] node[el,swap] {$c_{i+1},0$} (notxi)
(notxi) edge[bend right] node[el] {$c_{i+1},2$} (sink2)
(notxi) edge[bend left] node[el, align=left]{$b_i,0$\\$c_i,0$} (xi)
(xn) edge node[swap,el,align=left]{$b_n,0$\\$c_n,0$} (notxn)
(notxn) edge node[el] {$A,2$} (sink22)
(sink2) edge[loop] node[el,swap] {$A,0$} (sink2)
(sink22) edge[loop] node[el,swap] {$A,0$} (sink22)
;
\end{tikzpicture}
\caption{Counter gadget.}
\label{fig:counter-gadget}
\end{center}
\end{figure*}
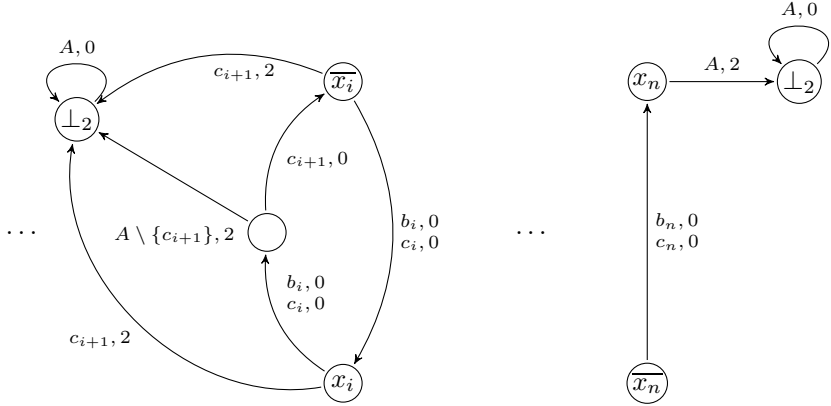

\begin{figure}
\begin{center}
\begin{tikzpicture}
\node[state] (v46) at (4,6) {};
\node[state] (v26) at (2,6) {};
\node[state] (v04) at (0,4.5) {};
\node[state] (v02) at (0,2) {};
\node[state] (v42) at (4,2) {};
\node[state] (v62) at (6,2) {};
\node[state] (v40) at (4,0) {};

\path
(v46) edge node[swap,el] {$b_0,0$} (v26)
(v46) edge node[el] {$b_0,0$} (v42)
(v26) edge node[swap,el] {$c_1,0$} (v04)
(v26) edge node[el] {$c_1,0$} (v42)
(v04) edge node[swap,el] {$c_2,0$} (v02)
(v04) edge node[el] {$c_2,0$} (v42)
(v02) edge node[el] {$c_3,0$} (v42)
(v02) edge[dotted,bend right] node[el,swap] {$c_3,0$} (v42)
(v42) edge node[el,swap] {$b_4,0$} (v40)
(v42) edge node[el] {$b_4,0$} (v62)
(v62) edge[dotted,bend left] node[el] {$c_5,0$} (v40)
(v62) edge[dotted,bend left] node[el,swap] {$c_5,0$} (v40)
;
\end{tikzpicture}
\caption{Adder gadget: depicted +$9$.}
\label{fig:adder-gadget}
\end{center}
\end{figure}
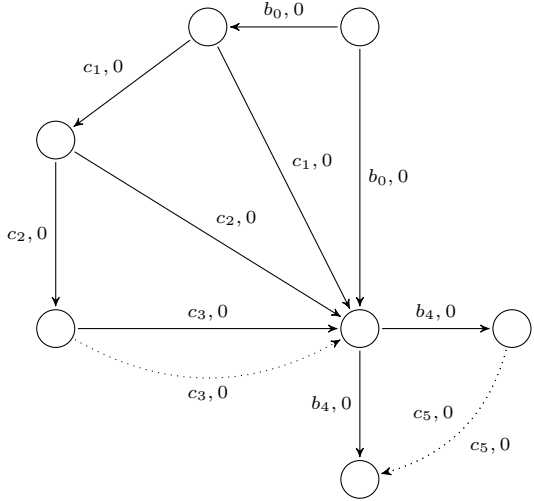

Our proof is by reduction from countdown games.  A countdown game $\calC$
consists of a weighted graph $(S,T)$, where $S$ is the set of states and $T
\subseteq S \times \mathbb{N} \setminus \{0\} \times S$ is the transition
relation, and a target value $N \in \mathbb{N}$. If $t = (s,d,s') \in T$ then we
say that the duration of the transition $t$ is $d$. A configuration of a
countdown game is a pair $(s,c)$, where $s \in S$ is a state and $c \in
\mathbb{N}$. A move of a countdown game from a configuration $(s,c)$ consists in
player Counter choosing a duration $d$ such that $(s,d,s') \in T$ for some $s'
\in S$ followed by player Spoiler choosing $s'$ such that $(s,d,s') \in T$, the
new configuration is then $(s', c + d)$. Counter wins if the game reaches a
configuration of the form $(s, N)$ and Spoiler wins if the game reaches a
configuration $(s,c)$ such that $c < N$ and for all $t = (s,d,s') \in T$ we have
that $c + d > N$.

Deciding the winner in a countdown game $\calC$ from a configuration
$(s,0)$---where $N$ and all durations in $\calC$ are given in binary---is
\EXP-complete~\cite{jsl08}.

Let us fix a countdown game $\calC = ((S,T),N)$ and let
$n = \lfloor \log_2 N \rfloor + 2$.

\myparagraph{Simplifying assumptions}	
Clearly, if Spoiler has a winning strategy and the game continues beyond
his winning the game, then eventually a configuration $(s,c)$, such that
$c \ge 2^n$, is reached. Thus, we can assume w.l.o.g.  that plays in
$\calC$ which visit a configuration $(s,N)$ are winning for Counter and
plays which don't visit a configuration $(s,N)$ but eventually get to a
configuration $(s',c)$ such that $c \ge 2^n$ are winning for Spoiler.

Additionally, we can also assume that $T$ in $\calC$ is total. That is
to say, for all $s \in S$ there is some duration $d$ such that $(s,d,s')
\in T$ for some $s' \in S$. If this were not the case then for every $s$
with no outgoing transitions we could add a transition
$(s,N+1,s_\bot)$ where $s_\bot$ is a newly added state. It is easy
to see that either player has a winning strategy in this new game if and
only if he has a winning strategy in the original game.

\myparagraph{Reduction}
We will now construct a sum-automaton $\calN$ with $\maxweight = 2$ such that,
in a regret game played on $\calN$, \eve
can ensure regret value strictly less than $2$ if and only if Counter
has a winning strategy in $\calC$. It will be clear from the proof how
to generalize the argument to any (strict or non-strict) regret
threshold.

The alphabet of the automaton $\calN =
(Q,\{q_I\},A,\Delta,w,F)$ is $A =
\{b_i \st 0 \le i \le n\} \cup \{c_i \st 0 < i \le n\} \cup
\{bail,choose\} \cup S$. We assume all states are final, \ie~$Q = F$.
We now describe the structure of $\calN$ (\ie
$Q$, $\Delta$ and $w$). 
	
\myparagraph{Initial gadget} Fig.~\ref{fig:initial-gadget} depicts the
initial state of the automaton.  Here, \eve has the choice of playing left
or right. If she plays to the left then \adam can play $bail$ and force
her to $\bot_0$ while the alternative run resulting from her having
chosen to go right goes to $\bot_2$. Hence, playing left already gives
\adam a winning strategy to ensure regret $2$, so she plays to the
right. If \adam now plays $bail$ then \eve can go to $\bot_2$ and as $W
=2$ this implies the regret will be $0$. Therefore, \adam plays $0$.

\myparagraph{Counter gadget} Fig.~\ref{fig:counter-gadget} shows the left
sub-automaton. All states from $\{\overline{x_i} \st 0 \le i \le n\}$ have
incoming transitions from the left part of the initial gadget with
symbol $A \setminus \{bail\}$ and weight $0$. Let $y_0 \dots y_n \in
\mathbb{B}$ be the (little-endian) binary representation of $N$, then
for all $x_i$ such that $y_i = 1$ there is a transition from $x_i$ to
$\bot_0$ with weight $0$ and symbol $bail$. Similarly, for all
$\overline{x_i}$ such that $y_i = 0$ there is a transition from
$\overline{x_i}$ to $\bot_0$ with weight $0$ and symbol $bail$. All the
remaining transitions not shown in the figure cycle on the same state,
e.g. $x_i$ goes to $x_i$ with symbol $choose$ and weight $0$. 

The sub-automaton we have just described corresponds to a counter gadget
(little-endian encoding) which keeps track of the sum of the durations
``spelled'' by \adam. At any point in time, the states of this sub-automaton
in which \eve \emph{believes} alternative runs are now will represent the
binary encoding of the current sum of durations.  Indeed, the initial
gadget makes sure \eve plays into the right sub-automaton and therefore she
knows there are alternative runs that could be at any of the
$\overline{x_i}$ states. This corresponds to the $0$ value of the
initial configuration.

\myparagraph{Adder gadget} Let us now focus on the right sub-automaton in which
\eve finds herself at the moment. The right transition with symbol
$A \setminus \{bail\}$ from the initial gadget goes to state $s$---the
initial state from $\calC$. It is easy to see how we can simulate
Counter's choice of duration and Spoiler's choice of successor. From $s$
there are transitions to every $(s,c)$, such that $(s,c,s') \in T$ for
some $s' \in S$ in $\calC$, with symbol $choose$ and weight $0$.
Transitions with all other symbols and weight $1$ going to $\bot_1$---a
trapping state with a $0$-weight cycle with every symbol---from $s$ ensure \adam
plays $choose$, else since $\maxweight = 2$ the regret of the game will be at
most $1$ and \eve wins.

Fig.~\ref{fig:adder-gadget} shows how \eve forces \adam to ``spell''
the duration $c$ of a transition of $\calC$ from $(s,c)$. For
concreteness, assume that \eve has chosen duration $9$. The top source
in Fig.~\ref{fig:adder-gadget} is therefore the state $(s,9)$. Again,
transitions with all the symbols not depicted go to $\bot_1$ with weight
$1$. Hence, \adam will play $b_0$ and \eve has the choice of going
straight down or moving to a state where \adam is forced to play $c_1$.
Recall from the description of the counter gadget that the belief of
\eve encodes the binary representation of the current sum of delays. If
she believes a play is in $x_1$ (and therefore none in $\overline{x_1}$)
then after \adam plays $b_0$ it is important for her to make him play
$c_1$ or this alternative run will end up in $\bot_2$. It will be clear
from the construction that \adam always has a strategy to keep the play
in the right sub-automaton without reaching $\bot_1$ and therefore if any
alternative run from the left sub-automaton is able to reach $\bot_2$ then
\adam wins (\ie can ensure regret $2$). Thus, \eve decides to force
\adam to play $c_1$. As the duration was $9$ this gadget now forces
\adam to play $b_4$ and again presents the choice of forcing \adam to
play $c_5$ to \eve. Clearly this can be generalized for any duration.
This gadget in fact simulates a \emph{cascade configuration} of $n$
$1$-bit adders.

Finally, from the bottom trap in the adder gadget, we have transitions
with symbols from $S$ with weight $0$ to the corresponding states (thus
simulating Spoiler's choice of successor state). Additionally, with any
symbol from $S$ and with weight $0$ \eve can also choose to go to a
state $q_{bail}$ where \adam is forced to play $bail$ and \eve is
forced into $\bot_0$.
	
\begin{proof}
	Note that if the simulation of the counter has
	been faithful and the belief of \eve encodes the value $N$ then by
	playing $bail$, \adam forces all of the alternative runs in the
	left sub-automaton into the $\bot_0$ trap. Hence, if Counter has a winning
	strategy and \eve faithfully simulates the $\calC$ she can force this
	outcome of all plays going to $\bot_0$. Note that from the right
	sub-automaton we have that $\bot_2$ is not reachable and therefore the
	highest value assigned to any word is $1$. Therefore, her regret is of
	at most $1$.

	Conversely, if both players faithfully simulate $\calC$ and the
	configuration $N$ is never reached, \ie Spoiler had a winning strategy
	in $\calC$, then eventually some alternative run in the left
	sub-automaton
	will reach $x_n$ and from there it will go to $\bot_2$. Again, the
	construction makes sure that \adam always has a strategy to keep the
	play in the right sub-automaton from reaching $\bot_1$ and therefore this
	outcome yields a regret of $2$ for \eve.
\end{proof}
\end{document}